\documentclass[11pt]{article}

\usepackage{mathrsfs}
\usepackage{amsmath,amssymb,amsfonts,amsthm}
\usepackage{moreverb,rotating,graphics,color,soul}
\usepackage[T1]{fontenc}
\usepackage[latin1]{inputenc}
\usepackage{epsfig}
\usepackage[francais,english]{babel} 
\usepackage{dsfont}

\newtheorem{theorem}{Theorem}
\newtheorem{proposition}[theorem]{Proposition}
\newtheorem{remark}[theorem]{Remark}
\newtheorem{lemma}[theorem]{Lemma}

\def\tr{{\rm Tr \,}}

\def\N{{\mathbb N}}

\def\R{{\mathbb R}}

\def\1{{\mathds{1}}}

\def\cA{{\mathcal A}}
\def\cB{{\mathcal B}}
\def\cC{{\mathcal C}}
\def\cE{{\mathcal E}}

\def\cH{{\mathcal H}}
\def\cK{{\mathcal K}}
\def\cL{{\mathcal L}}
\def\cO{{\mathcal O}}
\def\cP{{\mathcal P}}
\def\cS{{\mathcal S}}

\def\cZ{{\mathcal Z}}

\def\gC{{\mathscr C}}
\def\gS{{\mathfrak S}}

\newcommand \dps{\displaystyle }

\title{A mathematical perspective on \\ density functional perturbation theory}
\author{Eric Canc\`es and Nahia Mourad \\ \footnotesize{
Universit\'e Paris-Est, CERMICS, 
Ecole des Ponts and INRIA,} \\ \footnotesize{6 \& 8 avenue Blaise Pascal, 77455 Marne-la-Vall\'ee Cedex 2, France.}}

\date{May 29, 2014}

\begin{document}
 
\selectlanguage{english}

\maketitle

\begin{abstract}  This article is concerned with the mathematical analysis of the perturbation method for extended Kohn-Sham models, in which fractional occupation numbers are allowed. All our results are established in the framework of the reduced Hartree-Fock (rHF) model, but our approach can be used to study other kinds of extended Kohn-Sham models, under some assumptions on the mathematical structure of the exchange-correlation functional. The classical results of Density Functional Perturbation Theory in the non-degenerate case (that is when the Fermi level is not a degenerate eigenvalue of the mean-field Hamiltonian) are formalized, and a proof of Wigner's $(2n+1)$ rule is provided. We then focus on the situation when the Fermi level is a degenerate eigenvalue of the rHF Hamiltonian, which had not been considered so far.
\end{abstract}
 
\section{Introduction}

Eigenvalue perturbation theory has a long history. Introduced by Rayleigh~\cite{Ray77} in the 1870's, it was used for the first time in quantum mechanics in an article by Schr\"odinger~\cite{Sch26} published in 1926. The mathematical study of the perturbation theory of self-adjoint operators was initiated by Rellich~\cite{Rel37} in 1937, and has been since then the matter of a large number of contributions in the mathematical literature~(see \cite{Kat66,Rel69,Sim91} and references therein).

Perturbation theory plays a key role in quantum chemistry, where it is used in particular to compute the response properties of molecular systems to external electromagnetic fields (polarizability, hyperpolarizability, magnetic susceptibility, NMR shielding tensor, optical rotation, ...). Unless the number $N$ of electrons in the molecular system under study is very small, it is not possible to solve numerically the $3N$-dimensional electronic Schr\"odinger equation. In the commonly used Hartree-Fock and Kohn-Sham models, the {\em linear $3N$-dimensional} electronic Schr\"odinger equation is approximated by a coupled system of $N$ {\em nonlinear $3$-dimensional} Schr\"odinger equations. The adaptation of the standard linear perturbation theory to the nonlinear setting of the Hartree-Fock model is called Coupled-Perturbed Hartree-Fock theory (CPHF) in the chemistry literature~\cite{McW92} (see also~\cite{CanLeB98} for a mathematical analysis). Its adaptation to the Kohn-Sham model is usually referred to as the Density Functional Perturbation Theory (DFPT)~\cite{Bar87,Gon95}. The term Coupled-Perturbed Kohn-Sham theory is also sometimes used.

The purpose of this article is to study, within the reduced Hartree-Fock (rHF) framework, the perturbations of the ground state energy, 
the ground state density matrix, and the ground state density of a molecular system, when a ``small'' external potential is turned on. 

In the case when the Fermi level $\epsilon_{\rm F}^0$ is not a degenerate eigenvalue of the mean-field Hamiltonian (see Section~\ref{sec:rHF} for a precise definition of these objects), the formalism of DFPT is well-known (see e.g.~\cite{DreGro90}). It has been used a huge number of publications in chemistry and physics, as well as in a few mathematical publications, e.g. \cite{CanLew10,ELu13}. On the other hand, the degenerate case has not been considered yet, to the best of our knowledge. An interesting feature of DFPT in the degenerate case is that, in contrast with the usual situation in linear perturbation theory, the perturbation does not, in general, split the degenerate eigenvalue; it shifts the Fermi level and modifies the natural occupation numbers at the Fermi level.

The article is organized as follows. In Section~\ref{sec:rHF}, we recall the basic properties of rHF ground states and establish some new results on the uniqueness of the ground state density matrix for a few special cases. The classical results of DFPT in the non-degenerate case are recalled in Section~\ref{sec:DFPT}, and a simple proof of Wigner's $(2n+1)$ rule is provided. This very important rule for applications allows one to compute the perturbation of the energy at the $(2n+1)^{\rm st}$ order from the perturbation of the density matrix at the $n^{\rm th}$ order only. In particular, the atomic forces (first-order perturbations of the energy) can be computed from the unperturbed density matrix (Wigner's rule for $n=0$), while hyperpolarizabilities of molecules (second and third-order perturbations of the energy) can be computed from the first-order perturbation of the density matrix (Wigner's rule for $n=1$). In Section~\ref{sec:bounded}, we investigate the situation when the Fermi level is a degenerate eigenvalue of the rHF Hamiltonian. We establish all our results in the rHF framework in the whole space $\R^3$, for a local potential $W$ with finite Coulomb energy. Extensions to other frameworks (Hartree-Fock and Kohn-Sham models, supercell with periodic boundary conditions, nonlocal potentials, Stark external potentials, ...) are discussed in Section~\ref{sec:extensions}.
The proofs of the technical results are postponed until Section~\ref{sec:proofs}.

\section{Some properties of the rHF model} 
\label{sec:rHF}

Throughout this article, we consider a reference (unperturbed) system of $N$ electrons subjected to an external potential $V$. For a molecular system containing $M$ nuclei, $V$ is given by
$$
\forall x \in \R^3, \quad V(x) = -\sum_{k=1}^M z_k v(x-R_k),
$$
where $z_k \in \N^\ast$ is the charge (in atomic units) and $R_k \in \R^3$ the position of the $k^{\rm th}$ nucleus. For point nuclei $v=|\cdot|^{-1}$, while for smeared nuclei $v=\mu \star |\cdot|^{-1}$, where $\mu \in C^\infty_{\rm c}(\R^3)$ is a non-negative radial function such that $\int_{\R^3}\mu=1$.

In the framework of the (extended) Kohn-Sham model~\cite{DreGro90}, the ground state energy of this reference system is obtained by minimizing an energy functional of the form
\begin{equation}\label{eq:energy_KS}
E^{\rm KS}(\gamma) :=  \tr\left( -\frac 12 \Delta \gamma\right) + \int_{\R^3} \rho_\gamma V + \frac 12 D(\rho_\gamma,\rho_\gamma)+E^{\rm xc}(\rho_\gamma)
\end{equation}
over the set 
$$
\cK_N:=\left\{ \gamma \in \cS(L^2(\R^3)) \; | \; 0 \le \gamma \le 1, \; \tr(\gamma)=N, \; \tr(-\Delta\gamma) < \infty \right\}
$$
of the admissible one-body density matrices. To simplify the notation, we omit the spin variable. In the above definition, $\cS(L^2(\R^3))$ denotes the space of the bounded self-adjoint operators on $L^2(\R^3)$, $0 \le \gamma \le 1$ means that the spectrum of $\gamma$ is included in the range $[0,1]$, and $\tr(-\Delta\gamma)$ is the usual notation for $\tr(|\nabla|\gamma|\nabla|)$, where $|\nabla|:=(-\Delta)^{1/2}$ is the square root of the positive self-adjoint operator $-\Delta$ on $L^2(\R^3)$. The function $\rho_\gamma : \R^3\rightarrow\R_+$ is the electronic density associated with the density matrix $\gamma$. Loosely speaking, $\rho_\gamma(x)=\gamma(x,x)$, where $\gamma(x,y)$ is the kernel of the operator $\gamma$. It holds
$$
\rho_\gamma \ge 0, \quad \int_{\R^3} \rho_\gamma = N, \quad \int_{\R^3} |\nabla\sqrt{\rho_\gamma}|^2 \le \tr(-\Delta \gamma)
$$
(Hoffmann-Ostenhof inequality~\cite{Hof77}) so that, in particular, $\rho_\gamma \in L^1(\R^3) \cap L^3(\R^3)$.
The first term in the right-hand side of~(\ref{eq:energy_KS}) is the Kohn-Sham kinetic energy functional, the second one models the interaction of the electrons with the external potential $V$, $D(\cdot,\cdot)$ is the Coulomb energy functional defined on $L^{6/5}(\R^3) \times L^{6/5}(\R^3)$ by
$$
D(f,g) := \int_{\R^3}\int_{\R^3} \frac{f(x) \, g(y)}{|x-y|} \, dx \, dy,
$$
and $E^{\rm xc}$ is the exchange-correlation functional. In the reduced Hartree-Fock (rHF) model (also sometimes called the Hartree model), the latter functional is taken identically equal to zero. In the Local Density Approximation (LDA), it is chosen equal to
\begin{equation}\label{eq:defXCLDA}
E^{\rm xc}_{\rm LDA}(\rho) := \int_{\R^3} e_{\rm xc}(\rho(x)) \, dx,
\end{equation}
where the function $e_{\rm xc} \, : \, \R_+ \mapsto \R_-$ is such that for all $\overline{\rho} \in \R_+$, the non-positive number $e_{\rm xc}(\overline{\rho})$ is (an approximation of) the exchange-correlation energy density of the homogeneous electron gas with constant density $\overline{\rho}$. It is known that for neutral or positively charged molecular systems, that is when $Z=\sum_{k=1}^M z_k\ge N$, the minimization problem
\begin{equation}\label{eq:minKSg}
E_0 := \inf \left\{ E^{\rm KS}(\gamma), \; \gamma \in {\cal K}_N \right\},
\end{equation}
has a ground state $\gamma_0$, for the rHF model~\cite{Sol91} ($E^{\rm xc}=0$), as well as for the Kohn-Sham LDA model~\cite{AnaCan09} ($E^{\rm xc}=E^{\rm xc}_{\rm LDA}$).

\medskip

This contribution aims at studying, in the rHF setting, the perturbations of the ground state energy $E_0$, of the ground state density matrix $\gamma_0$, and of the ground state density $\rho_0=\rho_{\gamma_0}$ induced by an external potential $W$. In order to deal with both the unperturbed and the perturbed problem using the same formalism, we introduce the functional
$$
E^{\rm rHF}(\gamma,W) := \tr\left( -\frac 12 \Delta \gamma\right) + \int_{\R^3} \rho_\gamma V + \frac 12 D(\rho_\gamma,\rho_\gamma) + \int_{\R^3} \rho_\gamma W,
$$
and the minimization problem
\begin{equation} 
\cE^{\rm rHF}(W) := \inf \left\{E^{\rm rHF}(\gamma,W), \; \gamma \in \cK_N \right\}. \label{eq:min_rHF}
\end{equation}
We restrict ourselves to a potential $W$ belonging to the space
$$
\cC':=\left\{ v \in L^6(\R^3) \; | \; \nabla v \in (L^2(\R^3))^3 \right\},
$$
which can be identified with the dual of the Coulomb space 
$$
\cC:=\left\{ \rho \in {\cal S}'(\R^3) \, | \, \widehat \rho \in L^1_{\rm loc}(\R^3), \, |\cdot|^{-1} \widehat \rho \in L^2(\R^3) \right\}
$$
of the charge distributions with finite Coulomb energy. Here, ${\cal S}'(\R^3)$ is the space of tempered distributions on $\R^3$ and $\widehat \rho$ is the Fourier transform of $\rho$ (we use the normalization condition for which the Fourier transform is an isometry of $L^2(\R^3)$). When $W \in \cC'$, the last term of the energy functional should be interpreted as
$$
\int_{\R^3} \rho_\gamma W = \int_{\R^3} \overline{\widehat{\rho_\gamma}(k)} \; \widehat W(k) \, dk.
$$ 
The right-hand side of the above equation is well-defined as the functions $k \mapsto |k|^{-1} \widehat\rho_\gamma(k)$ and $k \mapsto |k|\widehat W(k)$ are both in $L^2(\R^3)$, since $\rho_\gamma \in L^1(\R^3) \cap L^3(\R^3) \subset L^{6/5}(\R^3) \subset \cC$.

\medskip

The reference, unperturbed, ground state is obtained by solving (\ref{eq:min_rHF}) with $W=0$. 

\medskip

\begin{theorem}[unperturbed ground state for the rHF model~\cite{Sol91}]\label{th:GS_rHF} If 
\begin{equation}\label{eq:NPC}
Z=\sum_{k=1}^M z_k \ge N \qquad \mbox{(neutral or positively charged molecular system)},
\end{equation}
then (\ref{eq:min_rHF}) has a ground state for $W=0$, and all the ground states share the same density $\rho_0$. The mean-field Hamiltonian
$$
H_{0} := -\frac 12 \Delta + V + \rho_0 \star |\cdot|^{-1},
$$
is a self-adjoint operator on $L^2(\R^3)$ and any ground state $\gamma_0$ is of the form
\begin{equation} \label{eq:ground_state}
\gamma_0 = \1_{(-\infty,\epsilon_{\rm F}^0)}(H_0) + \delta_0,
\end{equation}
with $\epsilon_{\rm F}^0 \le 0$, $0 \le \delta_0 \le 1$, $\mbox{\rm Ran}(\delta_0) \subset \mbox{\rm Ker}(H_0-\epsilon_{\rm F}^0)$.
\end{theorem}

The real number $\epsilon_{\rm F}^0$, called the Fermi level, can be interpreted as the Lagrange multiplier of the constraint $\tr(\gamma)=N$. The Hamiltonian $H_0$ is a self-adjoint operator on $L^2(\R^3)$ with domain $H^2(\R^3)$ and form domain $H^1(\R^3)$. Its essential spectrum is the range $[0,+\infty)$ and it possesses at least $N$ non-positive eigenvalues, counting multiplicities. For each $j \in \N^\ast$, we set 
$$
\epsilon_j := \inf_{X_j \subset {\cal X}_j} \; \sup_{v \in X_j, \, \|v\|_{L^2}=1} \langle v |H_0|v \rangle,
$$
where ${\cal X}_j$ is the set of the vector subspaces of $H^1(\R^3)$ of dimension $j$, and $v \mapsto \langle v |H_0|v \rangle$ the quadratic form associated with $H_0$. Recall (see e.g.~\cite[Section~XIII.1]{ReeSim78}) that $(\epsilon_j)_{j \in \N^\ast}$ is a non-decreasing sequence of real numbers converging to zero, and that, if $\epsilon_j$ is negative, then $H_0$ possesses at least $j$ negative eigenvalues (counting multiplicities) and $\epsilon_j$ is the $j^{\rm th}$ eigenvalue of $H_0$. We denote by $\phi_1^0, \phi_2^0, \cdots$ an orthonormal family of eigenvectors associated with the non-positive eigenvalues $\epsilon_1 \le \epsilon_2 \le \cdots$ of $H_0$. Three situations can {\it a priori} be encountered:
\begin{itemize}
\item {\bf Case 1 (non-degenerate case):} 
\begin{equation}\label{eq:C1}
\mbox{$H_0$ has at least $N$ negative eigenvalues and $\epsilon_N < \epsilon_{N+1} \le 0$.}
\end{equation}
In this case, the Fermi level $\epsilon_{\rm F}^0$ can be chosen equal to any real number in the range $(\epsilon_N,\epsilon_{N+1})$ and the ground state $\gamma_0$ is unique:
$$
\gamma_0 = \1_{(-\infty,\epsilon_{\rm F}^0)}(H_{\rho_0})  = \sum_{i=1}^N |\phi_i^0\rangle \langle\phi_i^0|;
$$ 
\item {\bf Case 2 (degenerate case):} 
\begin{equation}\label{eq:C2}
\mbox{$H_0$ has at least $N+1$ negative eigenvalues and $\epsilon_{N+1}=\epsilon_N$.}
\end{equation}
In this case, $\epsilon_{\rm F}^0=\epsilon_N=\epsilon_{N+1}< 0$;
\item {\bf Case 3 (singular case):} $\epsilon_{\rm F}^0=\epsilon_N=0$.
\end{itemize}

In the non-degenerate case, problem (\ref{eq:min_rHF}), for $W \in \cC'$ small enough, falls into the scope of the usual perturbation theory of nonlinear mean-field models dealt with in Section~\ref{sec:DFPT}. The main purpose of this article is to extend the perturbation theory to the degenerate case. We will leave aside the singular case $\epsilon_N=0$. It should be emphasized that the terminology {\em degenerate vs non-degenerate} used throughout this article refers to the possible degeneracy of the Fermi level, that is of a specific eigenvalue of the unperturbed mean-field Hamiltonian~$H_{\rho_0}$, not to the possible degeneracy of the Hessian of the unperturbed energy functional at $\gamma_0$. The perturbation method heavily relies on the uniqueness of the ground state density matrix $\gamma_0$ and on the invertibility of the Hessian (or more precisely of a reduced Hessian taking the constraints into account). In the non-degenerate case (Case~1), the minimizer $\gamma_0$ is unique and the reduced Hessian is always invertible. We will see that the same holds true in the degenerate case (Case~2) under assumption~(\ref{eq:hyp_ND}) below. We denote by
$$
N_{\rm f} := \mbox{Rank}\left( \1_{(-\infty,\epsilon_{\rm F}^0)}(H_0) \right)
$$
the number of (fully occupied) eigenvalues lower than $\epsilon_{\rm F}^0$, and by
$$
N_{\rm p} := \mbox{Rank}\left( \1_{\left\{\epsilon_{\rm F}^0\right\}}(H_0) \right)
$$
the number of (partially occupied) bound states of $H_0$ with energy $\epsilon_{\rm F}^0$. We also denote by $\R^{N_{\rm p} \times N_{\rm p}}_{\rm S}$ the space of real symmetric matrices of size $N_{\rm p}\times N_{\rm p}$.

\medskip

\begin{lemma}\label{lem:uniqueness}  Assume that (\ref{eq:NPC}) and (\ref{eq:C2}) are satisfied. If for any $M \in \R^{N_{\rm p} \times N_{\rm p}}_{\rm S}$,
\begin{equation}\label{eq:hyp_ND}
\left(\forall x \in \R^3, \;  \sum_{i,j=1}^{N_{\rm p}} M_{ij}\phi_{N_{\rm f}+i}^0(x)\phi_{N_{\rm f}+j}^0(x)  = 0 \right) \; \Rightarrow \; M=0, 
\end{equation}
then the ground state $\gamma_0$ of (\ref{eq:min_rHF}) for $W=0$ is unique
\end{lemma}

\medskip

The sufficient condition (\ref{eq:hyp_ND}) is satisfied in the following cases.

\medskip

\begin{proposition} \label{prop:uniqueness} Assume that (\ref{eq:NPC}) and (\ref{eq:C2}) are satisfied. If at least one of the two conditions below is fulfilled:
\begin{enumerate}
\item $N_{\rm p} \le 3$,
\item the external potential $V$ is radial and the degeneracy of $\epsilon_{\rm F}^0$ is essential,
\end{enumerate}
then (\ref{eq:hyp_ND}) holds true, and the ground state $\gamma_0$ of (\ref{eq:min_rHF}) for $W=0$ is therefore unique.
\end{proposition}

\medskip

Let us clarify the meaning of the second condition in Proposition~\ref{prop:uniqueness}. When $V$ is radial, the ground state density is radial, so that $H_0$ is a Schr\"odinger operator with radial potential:
$$
H_0=-\frac 12 \Delta + v(|x|).
$$
It is well-known (see e.g.~\cite[Section~XIII.3.B]{ReeSim78}) that all the eigenvalues of $H_0$ can be obtained by computing the eigenvalues of the one-dimensional Hamiltonians $h_{0,l}$, $l \in \N$, where $h_{0,l}$ is the self-adjoint operator on $L^2(0,+\infty)$ with domain $H^2(0,+\infty) \cap H^1_0(0,+\infty)$ defined by
$$
h_{0,l} := - \frac 12 \frac{d^2}{dr^2} + \frac{l(l+1)}{2r^2} + v(r).
$$
If $\epsilon_{\rm F}^0$ is an eigenvalue of $h_{0,l}$, then its multiplicity, as an eigenvalue of $H_0$, is at least $2l+1$. It is therefore degenerate as soon as $l \ge 1$. If $\epsilon_{\rm F}^0$ is an eigenvalue of no other $h_{0,l'}$, $l'\neq l$, then its multiplicity is exactly $2l+1$, and the degeneracy is called essential. Otherwise, the degeneracy is called accidental. It is well-known that for the very special case when $v(r)=-Zr^{-1}$ (hydrogen-like atom), accidental degeneracy occurs at every eigenvalue but the lowest one, which is non-degenerate. On the other hand, this phenomenon is really exceptional, and numerical simulations seem to show that, as expected, there is no accidental degeneracy at the Fermi level when $v$ is equal to the rHF mean-field potential of an atom (see~\cite{Mou13}).

\section{Density functional perturbation theory (non-degenerate case)}
\label{sec:DFPT}

We denote by ${\cal B}(X,Y)$ the space of bounded linear operators from the Banach space $X$ to the Banach space $Y$ (with, as usual, ${\cal B}(X):={\cal B}(X,X)$), by ${\cal S}(X)$ the space of self-adjoint operators on the Hilbert space $X$, by $\gS_1$ the space of trace class operators on $L^2(\R^3)$, and by $\gS_2$ the space of Hilbert-Schmidt operators on $L^2(\R^3)$ (all these spaces being endowed with their usual norms~\cite{ReeSim80,Sim79}). We also introduce the Banach space
$$
\gS_{1,1}:= \left\{T \in \gS_1 \; | \; |\nabla|T|\nabla| \in \gS_1 \right\},
$$
with norm
$$
\|T\|_{\gS_{1,1}}:=\|T\|_{\gS_1}+\||\nabla \, |T|\nabla| \, \|_{\gS_1}.
$$
We denote by $B_\eta(\cH)$ the open ball with center $0$ and radius $\eta > 0$ of the Hilbert space $\cH$.  

\medskip

Let us recall that in the non-degenerate case, 
$$
\gamma_0 \in \cP_N:= \left\{ \gamma \in {\cal S}(L^2(\R^3)) \; | \; \gamma^2=\gamma, \; \tr(\gamma)=N, \; \tr(-\Delta\gamma) < \infty\right\},
$$
that is $\gamma_0$ is a rank-$N$ orthogonal projector on $L^2(\R^3)$ with range in $H^1(\R^3)$, and
$$
\gamma_0 = \1_{(-\infty,\epsilon_{\rm F}^0]}(H_0) = \frac{1}{2i\pi} \oint_{\gC} (z-H_0)^{-1} \, dz,
$$
where $\gC$ is (for instance) the circle of the complex plane symmetric with respect to the real axis and intersecting it at points $\epsilon_1 -1$ and $\epsilon_{\rm F}^0$.

\subsection{Density matrix formulation}

The linear and multilinear maps introduced in the following lemma will be useful to write down the Rayleigh-Schr\"odinger expansions in compact forms.

\medskip

\begin{lemma} \label{lem:linear_response} Assume that (\ref{eq:NPC}) and (\ref{eq:C1}) are satisfied.
  \begin{enumerate}
  \item For each $k \in \N^\ast$, the $k$-linear map
\begin{eqnarray*}
Q^{(k)}: \qquad (\cC')^k \quad & \rightarrow & \gS_{1,1} \\
(v_1,\cdots,v_k) &\mapsto& \frac{1}{2i\pi} \oint_{\gC} (z-H_0)^{-1} v_1 (z-H_0)^{-1} v_2 \cdots (z-H_0)^{-1} v_k(z-H_0)^{-1} \, dz
\end{eqnarray*}
is well-defined and continuous. 

$\mbox{\rm Rank}(Q^{(k)}(v_1,\cdots,v_k)) \le N$ and $\tr(Q^{(k)}(v_1,\cdots,v_k))=0$, for all $(v_1,\cdots,v_k)~\in~(\cC')^k$, and there exists $0  < \alpha,C < \infty$ such that for all $k \in \N^\ast$ and all $(v_1,\cdots,v_k) \in (\cC')^k$,
\begin{equation}\label{eq:bounds_Qk}
\| Q^{(k)}(v_1,\cdots,v_k) \|_{\gS_{1,1}} \le C \alpha^k \|v_1\|_{\cC'} \cdots \|v_k\|_{\cC'}.
\end{equation}
\item The linear map 
\begin{eqnarray*}
\cL : \cC  & \rightarrow & \cC \\
\rho &\mapsto& -\rho_{Q^{(1)}(\rho \star |\cdot|^{-1})},
\end{eqnarray*}
associating to a charge density $\rho \in \cC$, minus the density $\rho_{Q^{(1)}(\rho \star |\cdot|^{-1})}$ of the trace-class operator $Q^{(1)}(\rho \star |\cdot|^{-1})$, is a bounded positive self-adjoint operator on $\cC$. As a consequence, $(1+\cL)$ is an invertible bounded positive self-adjoint operator on $\cC$.
  \end{enumerate}
\end{lemma}

\medskip

The main results of non-degenerate rHF perturbation theory for finite systems are gathered in the following theorem. 

\medskip

\begin{theorem}[rHF perturbation theory in the non-degenerate case] \label{th:CPrHF} Assume that (\ref{eq:NPC}) and (\ref{eq:C1}) are satisfied. Then, there exists $\eta > 0$ such that 
  \begin{enumerate}
  \item for all $W \in B_\eta(\cC')$, (\ref{eq:min_rHF}) has a unique minimizer $\gamma_W$. In addition, $\gamma_W \in \cP_N$ and 
\begin{equation}\label{eq:gammaWint}
\gamma_W = \1_{(-\infty,\epsilon_{\rm F}^0]}(H_W) = \frac{1}{2i\pi} \oint_{\gC} (z-H_W)^{-1} \, dz,
\end{equation}  
where 
$$
H_W = -\frac 12 \Delta + V + \rho_W \star |\cdot|^{-1}+W,
$$
$\rho_W$ being the density of $\gamma_W$;
\item the mappings $W \mapsto \gamma_W$, $W \mapsto \rho_W$ and $W \mapsto \cE^{\rm rHF}(W)$ are real analytic from $B_\eta(\cC')$ into $\gS_{1,1}$, $\cC$ and $\R$ respectively;
\item for all $W \in \cC'$ and all $-\eta \|W\|_{\cC'}^{-1} < \beta < \eta \|W\|_{\cC'}^{-1} $,
$$
\gamma_{\beta W} = \gamma_0 + \sum_{k=1}^{+\infty} \beta^k \gamma^{(k)}_W, \quad \rho_{\beta W} = \rho_0 +\sum_{k=1}^{+\infty} \beta^k \rho^{(k)}_W, \quad \cE^{\rm rHF}(\beta W) = \cE(0) + \sum_{k=1}^{+\infty} \beta^k \cE^{(k)}_W,
$$
the series being normally convergent in $\gS_{1,1}$, $\cC$ and $\R$ respectively;
\item  denoting by $W^{(1)}=W+\rho^{(1)}_W \star |\cdot|^{-1}$ and $W^{(k)}=\rho^{(k)}_W \star |\cdot|^{-1}$ for $k \ge 2$, the coefficients $\rho^{(k)}_W$ of the expansion of $\rho_{\beta W}$ can be obtained by the recursion relation
\begin{equation}\label{eq:CPrHF_rho}
(1+\cL)\rho_W^{(k)} = \widetilde \rho_W^{(k)},
\end{equation}
where $\widetilde \rho_W^{(k)}$ is the density of the operator $\widetilde Q_W^{(k)}$ defined by
\begin{equation}\label{eq:QWk}
\begin{array}{l} \dps
\widetilde Q_W^{(1)}=Q^{(1)}(W), \\ \dps \forall k \ge 2, \quad \widetilde Q_W^{(k)} =\sum_{l=2}^{k} \sum_{\tiny \begin{array}{c} 1 \le j_1,\cdots,j_l \le k-1, \\ \sum_{i=1}^l j_i=k
\end{array}}Q^{(l)}(W^{(j_1)},\cdots , W^{(j_{l})}); \end{array} 
\end{equation}
\item the coefficients $\gamma^{(k)}_W$ and $\cE^{(k)}_W$ are then given by
\begin{equation}\label{eq:gammaWk}
\gamma^{(k)}_W =  \frac{1}{2i\pi} \oint_{\gC} (z-H_0)^{-1} W^{(k)}(z-H_0)^{-1} \, dz
+ \widetilde Q_W^{(k)},
\end{equation}
and
\begin{equation}\label{eq:EWk2}
\cE^{(k)}_W = \tr\left(H_0\gamma^{(k)}_W\right)+\frac 12 \sum_{l=1}^{k-1} D\left(\rho^{(l)}_W,\rho^{(k-l)}_W\right) + \int_{\R^3} \rho^{(k-1)}_W \, W.
\end{equation}
  \end{enumerate}
\end{theorem}

\subsection{Molecular orbital formulation}

When $\epsilon_1 < \epsilon_2 < \cdots < \epsilon_N < \epsilon_{\rm F}^0$, that is when the lowest $N$ eigenvalues of $H_0$ are all non-degenerate, it can be seen, following the same lines as in~\cite{CanLeB98}, that, for all $W \in \cC'$, there exist real analytic functions $\beta \mapsto \epsilon_{W,i}(\beta) \in \R$ and $\beta \mapsto \phi_{W,i}(\beta) \in H^2(\R^3)$ defined in the neighborhood of $0$ such that $\epsilon_{W,i}(0)=\epsilon_i$, $\phi_{W,i}(0)=\phi_i^0$, and
$$
\left\{ \begin{array}{l} H_{\beta W} \phi_{W,i}(\beta) = \epsilon_{W,i}(\beta) \phi_{W,i}(\beta), \\
(\phi_{W,i}(\beta),\phi_{W,j}(\beta))_{L^2}=\delta_{ij}, \\ 
\epsilon_{W,1}(\beta) < \epsilon_{W,2}(\beta) <  \cdots < \epsilon_{W,N}(\beta) \mbox{ are the lowest eigenvalues of $H_{\beta W}$ (counting multiplicities)}.
\end{array}\right.
$$
The coefficients of the Rayleigh-Schr\"odinger expansions 
$$
\epsilon_{W,i}(\beta) =  \sum_{k=0}^{+\infty}\beta^k \epsilon_{W,i}^{(k)}, \qquad \phi_{W,i}(\beta) = \sum_{k=0}^{+\infty}\beta^k \phi_{W,i}^{(k)},
$$
where $\epsilon_{W,i}^0=\epsilon_i$ and $\phi_{W,i}^0=\phi_i^0$, are obtained by solving the system 
\begin{equation}\label{eq:RS_MO}
\forall k \in \N^\ast, \quad \forall 1 \le i \le N, \qquad \left\{ \begin{array}{l}
\dps \left( H_0 - \epsilon_i \right) \phi_{W,i}^{(k)} + \sum_{j=1}^N K_{ij}^0 \phi_{W,j}^{(k)} =  f_{W,i}^{(k)} + \epsilon_{W,i}^{(k)} \phi_i^0, \\
\dps \int_{\R^3} \phi_{W,i}^{(k)} \phi_i^0 = \alpha_{W,i}^{(k)},
\end{array} \right.
\end{equation}
where
$$
\forall \phi  \in L^2(\R^3), \quad K_{ij}^0 \phi = 2 \left(\phi_j^0\phi \star |\cdot|^{-1}\right)\phi_i^0,
$$
and where the right-hand sides  
$$
f_{W,i}^{(k)}= - W \phi_{W,i}^{(k-1)} - \sum_{j=1}^N   
\sum_{\tiny\begin{array}{c} 1 \le l_1,l_2,l_3 \le k-1, \\ l_1+l_2+l_3=k \end{array}}  
\left( \phi_{W,j}^{(l_1)}\phi_{W,j}^{(l_2)} \star |\cdot|^{-1} \right) \phi_{W,i}^{(l_3)}  + \sum_{l=1}^{k-1} \epsilon_{W,i}^{(l)}\phi_{W,i}^{(k-l)}, 
$$
and
$$
\alpha_{W,i}^{(k)}= - \frac 12 \sum_{l=1}^{k-1} \int_{\R^3}\phi_{W,i}^{(l)}\phi_{W,i}^{(k-l)} .
$$
at order $k$ only depend on the coefficients $\phi_{W,j}^{(l)}$ and $\epsilon_{W,j}^{(l)}$ at order $l \le k-1$. System (\ref{eq:RS_MO}) can therefore be considered as an infinite triangular system with respect to $k$.

\medskip

The fact that all the terms of the Rayleigh-Schr\"odinger series are defined unambiguously by (\ref{eq:RS_MO}) is guaranteed by the following lemma and the fact that for all $\phi$ and $\psi$ in $H^1(\R^3)$, $W\phi \in H^{-1}(\R^3)$ and $\phi\psi \star |\cdot|^{-1} \in L^\infty(\R^3)$.

\medskip

\begin{lemma}\label{lem:RS_MO} Assume that (\ref{eq:NPC}) and (\ref{eq:C1}) are satisfied and that $\epsilon_1 < \epsilon_2 < \cdots < \epsilon_N < \epsilon_{\rm F}^0$. For all $f=(f_1,\cdots,f_N) \in (H^{-1}(\R^3))^N$ and all $\alpha=(\alpha_1,\cdots,\alpha_N)\in \R^N$, the linear problem
\begin{equation}\label{eq:RS_MO_LS}
\forall 1 \le i \le N, \qquad \left\{ \begin{array}{l}
\dps \left( H_0 - \epsilon_i \right) \psi_i + \sum_{j=1}^N K_{ij}^0 \psi_j =  f_i + \eta_i \phi_i^0, \\
\dps \int_{\R^3} \psi_i \phi_i^0 = \alpha_i,
\end{array} \right.
\end{equation}
has a unique solution $(\Psi,\eta)=((\psi_1,\cdots,\psi_N),(\eta_1,\cdots,\eta_N))$ in $(H^1(\R^3))^N \times \R^N$. Moreover,  if $f\in (L^2(\R^3))^N$, then $\Psi \in (H^2(\R^3))^N$.
\end{lemma}

\medskip

Let us notice that, although the constraints $\int_{\R^3} \phi_{W,i}(\beta)\phi_{W,j}(\beta)=0$ for $i \neq j$ are not explicitly taken into account in the formal derivation of (\ref{eq:RS_MO}), the unique solution to (\ref{eq:RS_MO}) is compatible with these constraints since it automatically satisfies
\begin{equation}\label{eq:orthogonality}
\forall k \in \N^\ast, \quad \forall 1\le i,j \le N, \qquad \int_{\R^3} \sum_{l=0}^k \phi_{W,i}^{(l)}\phi_{W,j}^{(k-l)} = 0.
\end{equation}
A proof of the above result is provided in Section~\ref{sec:orthogonality}, together with the proof of Lemma~\ref{lem:RS_MO}.

\medskip

Let us finally mention that the Rayleigh-Schr\"odinger expansions of the density matrix $\gamma_{\beta W}$ and of the molecular orbitals $\phi_{W,i}(\beta)$ are related by
$$
\gamma_{W}^{(k)}=\sum_{i=1}^N \sum_{l=0}^k |\phi_{W,i}^{(l)}\rangle\langle\phi_{W,i}^{(k-l)}|,
$$
where we have used Dirac's bra-ket notation.

\subsection{Wigner's $(2n+1)$-rule}

According to~(\ref{eq:EWk2}), the first $n$ coefficients of the Rayleigh-Schr\"odinger expansion of the density matrix allows one to compute the first $n$ coefficients of the perturbation expansion of the energy. Wigner's ($2n+1$)-rule ensures that, in fact, they provide an approximation of the energy up to order $(2n+1)$. This property, which is very classical in linear perturbation theory, has been extended only recently to the nonlinear DFT framework~\cite{Ang09}. In the present section, we complement the results established in~\cite{Ang09} by providing a different, more general and compact proof, which also works in the infinite dimensional setting.

In the density matrix formulation, the Wigner's ($2n+1$)-rule can be formulated as follows. We introduce the nonlinear projector $\Pi$ on ${\cal S}(L^2(\R^3))$ defined by
$$
\forall T \in {\cal S}(L^2(\R^3)), \quad \Pi (T) = \1_{[1/2,+\infty)}(T),
$$
and, for $W \in \cC'$ and $\beta \in \R$, we denote by
$$
\widetilde \gamma_W^{(n)}(\beta):= \Pi\left( \gamma_0 + \sum_{k=1}^n \beta^k \gamma_W^{(k)} \right).
$$
For $T \in {\cal B}(L^2(\R^3))$, resp. $T \in \gS_2$, we denote by
$$
\mbox{\rm dist}(T,\cP_N):= \inf \left\{\|T-\gamma\|, \; \gamma \in {\cal P}_N\right\},
$$
resp.
$$
\mbox{\rm dist}_{\gS_2}(T,\cP_N):= \inf \left\{\|T-\gamma\|_{\gS_2}, \; \gamma \in {\cal P}_N\right\},
$$
the distance from $T$ to $\cP_N$ for the operator, resp. Hilbert-Schmidt, norm. The projector $\Pi$ enjoys the following properties.  

\medskip

\begin{lemma}\label{lem:projector_Pi} For each $T \in \Omega:=\left\{ T \in {\cal S}(L^2(\R^3)) \; | \; \mbox{\rm dist}(T,{\cal P}_N) < 1/2, \, \mbox{\rm Ran}(T) \subset H^1(\R^3)  \right\}$, $\Pi(T) \in {\cal P}_N$. Besides, for each $T \in \Omega \cap \gS_2$, $\Pi(T)$ is the unique solution to the variational problem
\begin{equation}\label{eq:projector_Pi}
\|T-\Pi(T)\|_{\gS_2} = \min_{\gamma \in {\cal P}_N} \|T-\gamma\|_{\gS_2} = \mbox{\rm dist}_{\gS_2}(T,{\cal P}_N).
\end{equation}
\end{lemma}

\medskip

It follows from Lemma~\ref{lem:projector_Pi} that, for all $W \in \cC'$ and $|\beta|$ small enough, $\widetilde \gamma_W^{(n)}(\beta)$ is the projection on ${\cal P}_N$ (in the sense of (\ref{eq:projector_Pi})) of the Rayleigh-Schr\"odinger expansion of the density matrix up to order $n$.

\medskip

\begin{theorem}[Wigner's ($2n+1$)-rule in the non-degenerate case] \label{th:Wigner_non_degenerate} 
Assume that (\ref{eq:NPC}) and (\ref{eq:C1}) are satisfied. For each $n \in \N$ and all $W \in \cC'$, it holds
\begin{equation}\label{eq:Wigner}
0 \le  E^{\rm rHF}(\widetilde \gamma_W^{(n)}(\beta),W) - \cE^{\rm rHF}(\beta W)  = {\cal O}(|\beta|^{2n+2}).
\end{equation}
\end{theorem}

Note that as $\gamma_0 + \sum_{k=1}^n \beta^k \gamma_W^{(k)}$ has finite-rank $N_n$, it can be diagonalized in an orthonormal basis of $L^2(\R^3)$ as
\begin{equation}\label{eq:gammaWn}
\gamma_0 + \sum_{k=1}^n \beta^k \gamma_W^{(k)}  = \sum_{i=1}^{N_n} g_{W,i}(\beta) |\widetilde\phi_{W,i}(\beta)\rangle\langle\widetilde\phi_{W,i}(\beta)|,
\end{equation}
with $(\widetilde\phi_{W,i}(\beta),\widetilde\phi_{W,j}(\beta))_{L^2}=\delta_{ij}$, $g_{W,i}(\beta) \in \R$, and $|g_{W,i}(\beta)| \ge |g_{W,j}(\beta)|$ for all $i \le j$. We also have 
$$
\sum_{i=1}^{N_n}g_{W,i}(\beta) = \tr\left(\gamma_0 + \sum_{k=1}^n \beta^k \gamma_W^{(k)}\right) = N,
$$
since, in view of (\ref{eq:QWk}), (\ref{eq:gammaWk}) and Lemma~\ref{lem:linear_response}, $\tr(\gamma_W^{(k)})=0$ for all $k \ge 1$. For $|\beta|$ small enough, the above operator is in $\Omega$, and therefore, $g_{W,1}(\beta) \ge g_{W,2}(\beta) \ge \cdots \ge g_{W,N}(\beta) > 1/2$ and $|g_{W,j}(\beta)|< 1/2$ for all $j \ge N+1$. We then have 
\begin{equation}\label{eq:tgammaWn}
\widetilde \gamma_W^{(n)}(\beta) = \sum_{i=1}^N |\widetilde\phi_{W,i}(\beta)\rangle\langle\widetilde\phi_{W,i}(\beta)|.
\end{equation}

\section{Perturbations of the rHF model in the degenerate case}
\label{sec:bounded}

We consider in this section the degenerate case. We assume that (\ref{eq:hyp_ND}) is satisfied, yielding that the ground state $\gamma_0$ of the unperturbed problem (\ref{eq:min_rHF}) with $W=0$ is unique. We also make the following assumption: 
\begin{equation}\label{eq:hyp_frac}
\epsilon_{\rm F}^0 < 0, \quad \mbox{Rank}(\delta_0)= N_{\rm p}, \quad \mbox{Ker}(1-\delta_0)=\left\{0\right\},
\end{equation}
where $\delta_0$ is the operator in (\ref{eq:ground_state}). Assumption (\ref{eq:hyp_frac}) means that the natural occupation numbers at the Fermi level (or in other words the $N_{\rm p}$ eigenvalues of $\delta_0|_{\rm \mbox{Ker}(H_0-\epsilon_{\rm F}^0)}$) are strictly comprised between $0$ and $1$. As a consequence, $\gamma_0$ belongs to the subset  
$$
{\cal K}_{N_{\rm f},N_{\rm p}}:=\left\{ \gamma \in {\cal K}_N \; | \; \mbox{Rank}(\gamma) = N_{\rm f}+N_{\rm p}, \; \mbox{Rank}(1-\gamma)=N_{\rm f} \right\}
$$
of ${\cal K}_N$. 

\medskip

We are going to prove that, under assumptions~(\ref{eq:hyp_ND}) and (\ref{eq:hyp_frac}), the rHF problem (\ref{eq:min_rHF}) has a unique minimizer for $\|W\|_{\cC'}$ small enough, which belongs to ${\cal K}_{N_{\rm f},N_{\rm p}}$ and whose dependence in $W$ is real analytic. To establish those results and compute the perturbation expansion in $W$ of the minimizer, we proceed as follow:
\begin{enumerate}
\item we first construct a real analytic local chart of ${\cal K}_{N_{\rm f},N_{\rm p}}$ in the vicinity of $\gamma_0$ (Section~\ref{sec:parametrization});
\item we use this local chart to prove that, for $\|W\|_{\cC'}$ small enough, the minimization problem
\begin{equation} 
\widetilde\cE^{\rm rHF}(W) := \inf \left\{E^{\rm rHF}(\gamma,W), \; \gamma \in \cK_{N_{\rm f},N_{\rm p}} \right\} \label{eq:min_rHF_KNN}
\end{equation}
has a unique local minimizer $\gamma_W$ in the vicinity of $\gamma_0$, and that the mappings $W \mapsto \gamma_W \in \gS_{1,1}$ and $W \mapsto \widetilde\cE^{\rm rHF}(W)$ are real analytic; we then prove that $\gamma_W$ is actually the unique global minimizer of~(\ref{eq:min_rHF}) (Section~\ref{sec:minKNN}), hence that $\widetilde\cE^{\rm rHF}(W)=\cE^{\rm rHF}(W)$;
\item we finally derive the coefficients of the Rayleigh-Schr\"odinger expansions of $\gamma_W$ and $\cE^{\rm rHF}(W)$, and prove that Wigner's $(2n+1)$-rule also holds true in the degenerate case (Section~\ref{sec:RSW}).
\end{enumerate}

\subsection{Parametrization of ${\cal K}_{N_{\rm f},N_{\rm p}}$ in the vicinity of $\gamma_0$} \label{sec:parametrization}

We first introduce the Hilbert spaces $\mathcal{H}_{\rm f}=\textmd{Ran}(\1_{(-\infty,\epsilon_{\rm F}^0)}(H_0))$, $\mathcal{H}_{\rm p}=\textmd{Ran}(\1_{\left\{\epsilon_{\rm F}^0\right\}}(H_0))$ and $\cH_{\rm u}=\textmd{Ran}(\1_{(\epsilon_{\rm F}^0,+\infty)}(H_0))$, corresponding respectively to the fully occupied, partially occupied, and unoccupied spaces of the unperturbed ground state density matrix $\gamma_0$. For later purpose, we also set $\mathcal{H}_{\rm o}=\mathcal{H}_{\rm f}\oplus\mathcal{H}_{\rm p}$. As 
$$
L^2(\R^3) = \mathcal{H}_{\rm f} \oplus \mathcal{H}_{\rm p} \oplus \mathcal{H}_{\rm u},
$$
any linear operator $T$ on $L^2(\R^3)$ can be written as a $3\times 3$ block operator 
$$
T = \left[\begin{matrix}
    T_{\rm ff}   &  T_{\rm fp}   &  T_{\rm fu} \\ \\   
    T_{\rm pf}   &  T_{\rm pp}   &  T_{\rm pu} \\ \\
    T_{\rm uf}   &  T_{\rm up}   &  T_{\rm uu}         
   \end{matrix}\right],
$$
where $T_{\rm xy}$ is a linear operator from $\mathcal{H}_{\rm y}$ to $\mathcal{H}_{\rm x}$. In particular, $\gamma_0$ and $H_0$ are block diagonal in this representation, and it holds
$$
\gamma_0 = \left[\begin{matrix}
    1 & 0 & 0  \\ \\   
    0  &  \Lambda   &  0 \\ \\
    0   &  0   &  0  
   \end{matrix}\right], \qquad 
H_0 = \left[\begin{matrix}
    H_0^{--} & 0 & 0  \\ \\   
    0  &  \epsilon_{\rm F}^0   &  0 \\ \\
    0   &  0   &  H_0^{++}  
   \end{matrix}\right]
$$
with $0 \le \Lambda = \delta_0|_{{\cal H}_p} \le 1$, $H_0^{--}-\epsilon_{\rm F}^0 \le - g_-:=\epsilon_{N_{\rm f}}-\epsilon_{\rm F}^0$ and  $H_0^{++}-\epsilon_{\rm F}^0 \ge  g_+:=\epsilon_{N_{\rm f}+N_{\rm p}+1}-\epsilon_{\rm F}^0$.

We then introduce
\begin{itemize}
\item the spaces of finite-rank operators
\begin{equation*}
 \mathcal{A}_{\rm ux}:=\left\{A_{\rm ux} \in {\cal B}(\mathcal{H}_{\rm x},\mathcal{H}_{\rm u}) \; | \;  (H_0^{++}-\epsilon_{\rm F}^0)^{1/2}A_{\rm ux}\in{\cal B}(\mathcal{H}_{\rm x},\mathcal{H}_{\rm u})\right\},
\end{equation*}
for ${\rm x}\in\{{\rm f},{\rm p}\}$, endowed with the inner product
$$
(A_{\rm ux},B_{\rm ux})_{\mathcal{A}_{\rm ux}}:=\tr( A_{\rm ux}^\ast(H_0^{++}-\epsilon_{\rm F}^0) B_{\rm ux});
$$
\item the finite dimensional spaces
$$
\mathcal{A}_{\rm pf}:={\cal B}(\mathcal{H}_{\rm f},\mathcal{H}_{\rm p})
$$
and
\begin{equation*}
\mathcal{A}_{\rm pp}:=\{A_{\rm pp}\in\mathcal{S}(\mathcal{H}_{\rm p})  \; | \; \tr(A_{\rm pp})=0\};
\end{equation*}
\item the product space
\begin{equation*}
 \mathcal{A}:=\mathcal{A}_{\rm uf}\times\mathcal{A}_{\rm up}\times\mathcal{A}_{\rm pf}\times\mathcal{A}_{\rm pp},
\end{equation*}
which we endow with the inner product
$$
(A,B)_\cA= \sum_{{\rm x} \in \{{\rm f},{\rm p}\}} (A_{\rm ux},B_{\rm ux})_{\cA_{\rm ux}} + \sum_{{\rm x} \in \{{\rm f},{\rm p}\}} \tr\left( A_{\rm px}B_{\rm px}^* \right). 
$$
\end{itemize}
To any $A=(A_{\rm uf},A_{\rm up},A_{\rm pf},A_{\rm pp}) \in {\cal A}$, we associate the bounded linear operator $\Gamma(A)$ on $L^2(\R^3)$ defined as 
\begin{equation}
\Gamma(A) := \exp\left(L_{\rm uo}(A)\right) \; \exp\left(L_{\rm pf}(A)\right) \; \left( \gamma_0+L_{\rm pp}(A)\right) \; \exp\left(-L_{\rm pf}(A)\right) \; \exp\left(-L_{\rm uo}(A)\right),
\end{equation}
where
$$
L_{\rm uo}(A) := \left[\begin{matrix}
    0 & 0 & -A_{\rm uf}^\ast  \\ \\   
    0  &  0   &   -A_{\rm up}^\ast  \\ \\
    A_{\rm uf}   &  A_{\rm up}   &  0  
   \end{matrix}\right] , \quad L_{\rm pf}(A) := \left[\begin{matrix}
    0 & -A_{\rm pf}^\ast & 0  \\ \\   
    A_{\rm pf}  &  0   &  0 \\ \\
    0   &  0   &  0  
   \end{matrix}\right] , \quad L_{\rm pp}(A) := \left[\begin{matrix}
    0 & 0 & 0  \\ \\   
    0  &  A_{\rm pp}   &  0 \\ \\
    0   &  0   &  0  
   \end{matrix}\right].
$$
Note that $\Gamma$ is real analytic from ${\cal A}$ to $\gS_{1,1}$, $\Gamma(0)=\gamma_0$, and $\Gamma(A) \in {\cal K}_N$ for all $A_{\rm pp}$ such that $0 \le \Lambda+A_{\rm pp} \le 1$. In addition, it follows from Assumption~(\ref{eq:hyp_frac}) that $\Gamma(A) \in {\cal K}_{N_{\rm f},N_{\rm p}}$ for all $A \in {\cal A}$ small enough. The following lemma provides the parametrization of ${\cal K}_{N_{\rm f},N_{\rm p}}$ near $\gamma_0$ our analysis is based upon.

\medskip

\begin{lemma}\label{lem:localmap} Assume that (\ref{eq:NPC}), (\ref{eq:C2}), (\ref{eq:hyp_ND}) and (\ref{eq:hyp_frac}) are satisfied. Then there exists an open neighborhood ${\cal O}$ of $0$ in ${\cal A}$ and an open neighborhood ${\cal O}'$ of $\gamma_0$ in $\gS_{1,1}$ such that the real analytic mapping
\begin{equation}\label{eq:localchart}
\begin{array}{ccl}
{\cal O} &\rightarrow& {\cal K}_{N_{\rm f},N_{\rm p}}\cap {\cal O}' \\
A &\mapsto& \Gamma(A) 
\end{array}
\end{equation}
is bijective.
\end{lemma}

\medskip

In other words, the inverse of the above mapping is a local chart of ${\cal K}_{N_{\rm f},N_{\rm p}}$ in the vicinity of $\gamma_0$. Note that a similar, though not identical, parametrization of the finite-dimensional counterpart of ${\cal K}_{N_{\rm f},N_{\rm p}}$ obtained by discretization in atomic orbital basis sets, was used in~\cite{CanKud03} to design quadratically convergent self-consistent algorithms for the extended Kohn-Sham model.

\subsection{Existence and uniqueness of the minimizer of (\ref{eq:min_rHF}) for $W$ small enough}
\label{sec:minKNN}

We now define the energy functional 
\begin{equation}\label{eq:EAW}
E(A,W) := E^{\rm rHF}(\Gamma(A),W),
\end{equation}
for all $A \in {\cal O}$ and all $W \in \cC'$, which, in view of Lemma~\ref{lem:localmap} allows us to study the existence and uniqueness of local minimizers of (\ref{eq:min_rHF_KNN}) in the vicinity of $\gamma_0$ when $\|W\|_{\cC'}$ is small enough. The functional $E$ is clearly real analytic; we denote by 
\begin{equation}\label{eq:FAW}
F(A,W):= \nabla_A E(A,W),
\end{equation}
the gradient of $E$ with respect to $A$, evaluated at point $(A,W)$. As $\gamma_0$ is the unique minimizer of the functional $\gamma \mapsto E^{\rm rHF}(\gamma,0)$ on ${\cal K}_N$, hence on ${\cal K}_{N_{\rm f},N_{\rm p}}$, $0$ is the unique minimizer of the functional $A \mapsto E(A,0)$ on ${\cal O}$, so that
$$
F(0,0) = 0.
$$ 

\medskip

\begin{lemma}\label{lem:preIFT} Assume that (\ref{eq:NPC}), (\ref{eq:C2}), (\ref{eq:hyp_ND}) and (\ref{eq:hyp_frac}) are satisfied. Let 
$$
\Theta:= \frac 12 F'_A(0,0)|_{{\cal A}\times \left\{0\right\}}, 
$$
where $F'_A(0,0)|_{{\cal A}\times \left\{0\right\}}$ is the restriction to the subspace ${\cal A} \times \left\{0\right\} \equiv {\cal A}$ of ${\cal A} \times \cC'$ of the derivative of $F$ with respect to $A$ at $(0,0)$. The linear map $\Theta$ is a bicontinuous coercive isomorphism from~${\cal A}$ to its dual~${\cal A}'$.
\end{lemma}

\medskip

We infer from Lemma~\ref{lem:preIFT} and the real analytic version of the implicit function theorem that for $W \in {\cal C}'$ small enough, the equation $F(A,W)=0$ has a unique solution $\widetilde A(W)$ in ${\cal O}$, and that the function $W \mapsto \widetilde A(W)$ is real analytic in the neighborhood of $0$. It readily follows from (\ref{eq:FAW}) and Lemma~\ref{lem:localmap} that for $W \in {\cal C}'$ small enough,
\begin{equation}\label{eq:def_gamma_W}
\gamma_W:=\Gamma(\widetilde A(W))
\end{equation}
is the unique critical point of (\ref{eq:min_rHF_KNN}) in the vicinity of $\gamma_0$. This critical point is in fact a local minimizer since $\Theta$, which is in fact the second derivative of the energy functional $A \mapsto E(A,0)$, is coercive. We have actually the following much stronger result.

\medskip

\begin{lemma}\label{lem:min_KNN=min_K} Assume that (\ref{eq:NPC}), (\ref{eq:C2}), (\ref{eq:hyp_ND}) and (\ref{eq:hyp_frac}) are satisfied. Then, for $\|W\|_{\cC'}$ small enough, the density matrix $\gamma_W$ defined by (\ref{eq:def_gamma_W}) is the unique global minimizer of (\ref{eq:min_rHF}).
\end{lemma}

\medskip

We conclude this section by providing the explicit form of $\Theta$, which is useful to prove Lemma~\ref{lem:preIFT}, but also to compute the Rayleigh-Schr\"odinger expansion of $\gamma_W$:
\begin{eqnarray*} \, 
[\Theta(A)]_{\rm uf} & = & - A_{\rm uf} (H_0^{--}-\epsilon_{\rm F}^0)+ (H_0^{++}-\epsilon_{\rm F}^0) A_{\rm uf} +\frac{1}{2}[{\cal J}(A)]_{\rm uf}, \\ \, 
[\Theta(A)]_{\rm up} & = & (H_0^{++}-\epsilon_{\rm F}^0)A_{\rm up}\Lambda+\frac{1}{2} [{\cal J}(A)]_{\rm up},  \\ \,
[\Theta(A)]_{\rm pf} & = & -(1-\Lambda)A_{\rm pf}\left(H_0^{--}-\epsilon_{\rm F}^0\right)+\frac{1}{2}[{\cal J}(A)]_{\rm pf}, \\
\,  [ \Theta (A) ]_{\rm pp} & = & \frac 12 [{\cal J}(A)]_{\rm pp},
\end{eqnarray*}
${\cal J}$ denoting the linear operator from ${\cal A}$ to ${\cal A}'$ defined by 
$$
\forall (A,A') \in {\cal A} \times {\cal A}, \quad \langle{\cal J}(A),A'\rangle = D(\rho_{\gamma_1(A)}, \rho_{\gamma_1(A')}), 
$$
where
\begin{equation}\label{eq:gamma1A}
\gamma_1(A)=\langle \Gamma'(0), A \rangle = \left[ L_{\rm uo}(A)+L_{\rm pf}(A),\gamma_0\right] + L_{\rm pp}(A).
\end{equation}
A key observation for the sequel is that  \begin{equation}\label{eq:trH0g1}
\forall A \in \cA, \quad \tr\left(H_0\gamma_1(A)\right)=0.
\end{equation}

\subsection{Rayleigh-Schr\"odinger expansions}
\label{sec:RSW}

It immediately follows from the previous two sections that, for any $W \in \cC'$, the functions $\beta \mapsto A_W(\beta):=\widetilde A(\beta W)$ and $\beta \mapsto \gamma_{\beta W}:=\Gamma(\widetilde A(\beta W))$ are well-defined and real analytic in the vicinity of $0$. The purpose of this section is to provide a method to compute the coefficients $A^{(k)}_W$, $\gamma^{(k)}_W$ and $\cE^{(k)}_W$ of the expansions
$$
A_W(\beta) = \sum_{k=1}^{+\infty} \beta^k A_W^{(k)}, \quad \gamma_{\beta W} = \gamma_0 + \sum_{k=1}^{+\infty} \beta^k \gamma^{(k)}_W \quad \mbox{and} \quad  \cE^{\rm rHF}(\beta W) = \cE^{\rm rHF}(0) + \sum_{k=1}^{+\infty} \beta^k \cE^{(k)}_W.
$$
We can already notice that the coefficients $\gamma^{(k)}_W$ and $\cE^{(k)}_W$ are easily deduced from the coefficients $A^{(k)}_W$. Using the following version of the Baker-Campbell-Hausdorff formula
$$
e^X Y e^{-X} = Y + [X,Y]+ \frac{1}{2!} [X,[X,Y]]+ \frac{1}{3!} [X,[X,[X,Y]]] + ...,
$$
we indeed obtain 
\begin{equation}\label{eq:gammak}
 \gamma^{(k)}_W=\sum_{1\leq l\leq k} \;\;\; \sum_{\alpha \in (\N^\ast)^l \, | \, |\alpha|_1=k} \gamma_{W,l}^{\alpha} \qquad \mbox{with} \qquad \gamma_{W,l}^{\alpha} =
\gamma_l(A_W^{(\alpha_1)}, \cdots, A_W^{(\alpha_l)}),
\end{equation}
where for all $\alpha=(\alpha_1,\cdots,\alpha_l) \in (\N^\ast)^l$, $|\alpha|_1=\alpha_1+\cdots+\alpha_l$, $|\alpha|_\infty=\max(\alpha_i)$, and 
\begin{eqnarray*}
\gamma_l(A_1, \cdots, A_l)
&=& \sum_{i+j=l} \frac{1}{i!j!}[L_{\rm uo}(A_{1}),\cdots,[L_{\rm uo}(A_{i}),[L_{\rm pf}(A_{{i+1}}),\cdots,[L_{\rm pf}(A_{l}),\gamma_0]\cdots] \\
            & + & \sum_{i+j= l-1}\frac{1}{i!j!} [L_{\rm uo}(A_{1}),...,[L_{\rm uo}(A_{i}),[L_{\rm pf}(A_{{i+1}}),\cdots,[L_{\rm pf}(A_{{l-1}}),L_{\rm pp}(A_{l})]\cdots],
\end{eqnarray*}
for all $(A_1,\cdots,A_l) \in {\cal A}^{l}$. Note that for $l=1$, the above definition agrees with (\ref{eq:gamma1A}), and that, more generally,
\begin{equation}\label{eq:expGA}
\forall A \in {\cal A}, \quad \Gamma(A) = \gamma_0+\sum_{l=1}^{+\infty} \gamma_l(A,\cdots,A).
\end{equation}

\medskip

\noindent
It follows from (\ref{eq:trH0g1}) and (\ref{eq:gammak}) that 
\begin{equation}\label{eq:EW1}
{\cal E}_W^{(1)} = \int_{\R^3} \rho_{\gamma_0} W,
\end{equation}
and that for all $k \ge 2$,
\begin{equation}\label{eq:EWk}
{\cal E}_W^{(k)} =  \tr\left(  -\frac 12 \Delta \gamma_W^{(k)} \right) + \int_{\R^3} \rho_{\gamma_W^{(k)}} V+ \frac 12 \sum_{l=0}^k D\left( \rho_{\gamma_W^{(l)}}, \rho_{\gamma_W^{(k-l)}}\right) + \int_{\R^3} \rho_{\gamma_W^{(k-1)}} W  
\end{equation}
We will see however that the above formula is far from being optimal, in the sense that ${\cal E}_W^{(k)}$ can be computed using the coefficients $A_W^{(j)}$ for $1 \le j \le k/2$ only (see formulation (\ref{eq:EWk1}) of Wigner's $(2n+1)$-rule), whereas the direct evaluation of ${\cal E}_W^{(k)}$ based on (\ref{eq:gammak}) and (\ref{eq:EWk}) requires the knowledge of the $A_W^{(j)}$'s up to $j=k$.

\subsection{Main results for the degenerate case}

The following theorem collects the results obtained so far, and provides a systematic way to construct the $A_W^{(k)}$'s, as well as an extension to Wigner's $(2n+1)$-rule to the degenerate case. 

\medskip

\begin{theorem}\label{th:rHF} Assume that (\ref{eq:NPC}), (\ref{eq:C2}), (\ref{eq:hyp_ND}) and (\ref{eq:hyp_frac}) are satisfied. Then there exists $\eta > 0$, such that 
  \begin{enumerate}
  \item {\rm existence and uniqueness of the ground state:} for all $W \in B_\eta(\cC')$, the rHF model (\ref{eq:min_rHF}) has a unique ground state $\gamma_W$; 
  \item {\rm no energy level splitting at the Fermi level:} the mean-field Hamiltonian 
  $$
  H_W = -\frac 12 \Delta + V + \rho_W \star |\cdot|^{-1}+W
  $$
(where $\rho_W$ is the density of $\gamma_W$) has at least $N_{\rm o}=N_{\rm f}+N_{\rm p}$ negative eigenvalues (counting multiplicities), the degeneracy of the $(N_{\rm f}+1)^{\rm st}$ eigenvalue, which is also the Fermi level $\epsilon_{\rm F}^W$ of the system, being equal to $N_{\rm p}$, and it holds
$$
\gamma_W = \1_{(-\infty,\epsilon_{\rm F}^W)}(H_W)+\delta_W,
$$
where $0 \le \delta_W \le 1$ is an operator such that ${\rm Ran}(\delta_W) \subset {\rm Ker}(H_W-\epsilon_{\rm F}^W)$ with maximal rank $N_{\rm p}$;
  \item {\rm analyticity of the ground state:} the functions $W \mapsto \gamma_W$ and $W \mapsto \cE^{\rm rHF}(W)$ are real analytic from $B_\eta(\cC')$ to $\gS_{1,1}$ and $\R$ respectively. 
  For all $W \in \cC'$ and all $-\eta \|W\|_{\cC'}^{-1} < \beta < \eta \|W\|_{\cC'}^{-1} $,
  $$
  \gamma_{\beta W} = \gamma_0 + \sum_{k=1}^{+\infty} \beta^k \gamma^{(k)}_W, \quad \cE^{\rm rHF}(\beta W) = \cE^{\rm rHF}(0) + \sum_{k=1}^{+\infty} \beta^k \cE^{(k)}_W,
  $$
  the series being normally convergent in $\gS_{1,1}$ and $\R$ respectively;
  \item {\rm Rayleigh-Schr\"odinger expansions:} the coefficients $\gamma_W^{(k)}$  are given  by (\ref{eq:gammak}), where the $A_W^{(k)}$'s are obtained recursively by solving the well-posed linear problem in ${\cal A}$
\begin{equation}\label{eq:linsysA}
\Theta (A_W^{(k)})= -\frac{1}{2} B_W^{(k)},
\end{equation}
where the $B_W^{(k)}$'s are defined by
\begin{equation}\label{eq:formulaB1}
\forall  A \in {\cal A}, \quad \langle B_W^{(1)},A \rangle = \int_{\R^3}\rho_{\gamma_1(A)}W,
\end{equation}
and for all $k \ge 2$ and all $A \in {\cal A}$,
\begin{eqnarray}\label{eq:formulaBk}
&&  \!\!\!\!\!\!\!\!\!\!\!\!
 \langle B_W^{(k)},A \rangle=\sum_{l=3}^{k+1} \;\;\;  \sum_{\underset{|\alpha|_1=k, \, |\alpha|_\infty\le k-1}{\alpha \in (\N^\ast)^{l-1}}} \sum_{i=1}^l \tr\left(H_0\gamma_l(\tau_{(i,l)}(A_W^{(\alpha_1)},\cdots,A_W^{(\alpha_{l-1})},A))\right)  \nonumber \\
 && \!\!\!\!\!\!\!\!\!\!\!\!+ \sum_{\underset{l\geq 1,\;l'\geq 1}{3 \le l+l' \le k+1}} \;\;\; \sum_{\underset{|\alpha|_1+|\alpha'|_1=k, \, \max(|\alpha|_\infty,|\alpha'|_\infty)\le k-1}{\alpha\in(\N^\ast)^l,\, \alpha'\in(\N^\ast)^{l'-1}}} \sum_{i=1}^{l'} D\left(\rho_{\gamma_{W,l}^{\alpha}},\rho_{\gamma_{l'}(\tau_{(i,l')}(A_W^{(\alpha'_1)},\cdots,A_W^{(\alpha'_{l'-1})},A))}\right) \nonumber \\ 
&& \!\!\!\!\!\!\!\!\!\!\!\! +\sum_{l=2}^{k}  \;\;\; \sum_{\underset{|\alpha|_1=k-1, \, |\alpha|_\infty\le k-2}{\alpha \in (\N^\ast)^{l-1}}} \,\sum_{i=1}^l \int_{\R^3}\rho_{\gamma_l(\tau_{(i,l)}(A_W^{(\alpha_1)},\cdots,A_W^{(\alpha_{l-1})},A))}W; 
\end{eqnarray}
 where $\tau_{(i,j)}$ is the transposition swapping the $i^{\rm th}$ and $j^{\rm th}$ terms (by convention $\tau_{(i,i)}$ is the identity);\item first formulation of Wigner's ($2n+1$)-rule:  for all $n\in \N$, and all $\epsilon \in \left\{0,1\right\}$,
\begin{eqnarray}
{\cal E}_W^{(2n+\epsilon)} &=&   \sum_{2\leq l\leq 2n+\epsilon} \;\;\;
    \sum_{\alpha \in (\N^\ast)^l \, | \, |\alpha|_1=2n+\epsilon,\, |\alpha|_\infty\leq n} \tr(H_0 \gamma_{W,l}^\alpha) \nonumber \\
&+& \frac{1}{2} \sum_{\underset{l,l'\geq 1}{2 \le l+l'\le 2n+\epsilon}} \;\;\; \sum_{\underset{\max(|\alpha|_\infty,\,|\alpha'|_\infty)\leq n}{\alpha \in (\N^\ast)^l, \, \alpha' \in (\N^\ast)^{l'} \, | \, |\alpha|_1+|\alpha'|_1 = 2n+\epsilon}} D\left(\rho_{\gamma_{W,l}^{\alpha}},\rho_{\gamma_{W,l'}^{\alpha'}}\right) \nonumber \\
    & +& \sum_{1\leq l\leq 2n+\epsilon-1}\;\;\;  \sum_{\alpha \in (\N^\ast)^l \, | \, |\alpha|_1 = 2n+\epsilon-1,\, |\alpha|_\infty\leq n} \int_{\R^3}\rho_{\gamma_{W,l}^{\alpha}} W; \label{eq:EWk1}
\end{eqnarray}

\item second formulation of Wigner's ($2n+1$)-rule: it holds
\begin{equation}\label{eq:Wigner_degenerate}
0 \le E^{\rm rHF}\left( \Gamma\left( \sum_{k=1}^n \beta^k A_W^{(k)} \right),W \right) - {\cal E}^{\rm rHF}(\beta W) = {\cal O}(|\beta|^{2n+2}).
\end{equation}
\end{enumerate}
\end{theorem}

Note that both formulations of Wigner's ($2n+1$)-rule state that an approximation of the energy ${\cal E}^{\rm rHF}(\beta W)$ up to order $(2n+1)$ in $\beta$, can be obtained from the $A_W^{(k)}$ for $1 \le k \le n$. They are yet different since the first formulation consists in computing all the coefficients $\cE_W^{(k)}$ up to order $(2n+1)$, while the second formulation is based on the computation of the density matrix $\Gamma\left( \sum_{k=1}^n \beta^k A_W^{(k)} \right)$.

\medskip

\begin{remark} Although we were not able to rigorously prove that assumptions  (\ref{eq:NPC}), (\ref{eq:C2}), (\ref{eq:hyp_ND}) and (\ref{eq:hyp_frac}) were actually satisfied for a specific molecular system, we strongly believe that this is the case for some atoms. Recall that the singlet-spin state rHF model is obtained from the spinless rHF model dealt with here by replacing $N$ by $N/2$ (the number of electron pairs) and $\rho_\gamma$ by $2\rho_\gamma$ (each state is occupied by one spin-up and one spin-down electron), so that all our results can be applied {\it mutatis mutandis} to the singlet-spin state rHF model. We have performed numerical simulations of a carbon atom within the singlet-spin state rHF model \cite{Mou13} and observed that for this system, the lowest two eigenvalues of $H_0$, corresponding to the 1s and 2s shells, are negative and non-degenerate, while the third lowest eigenvalue, corresponding to the 2p shell, is threefold degenerate. As the carbon atom contains six electrons, that is three electron pairs, the Fermi level coincides with the third lowest eigenvalue. Using the first statement of Proposition~\ref{prop:uniqueness}, we obtain that assumptions (\ref{eq:C2}) is satisfied, hence that the ground state density matrix $\gamma_0$ is unique, yielding that, by symmetry, all the occupation numbers at the Fermi level are equal to $1/3$. Numerical simulations therefore suggest that assumptions (\ref{eq:C2}), (\ref{eq:hyp_ND}) and (\ref{eq:hyp_frac}) are satisfied for the singlet-spin state rHF model of a carbon atom, while (\ref{eq:NPC}) is obviously satisfied since this system is electrically neutral.
\end{remark}

\medskip

\begin{remark} In order to illustrate what may happen when assumption~(\ref{eq:hyp_frac}) is not satisfied, we consider the toy model 
\begin{equation}\label{eq:toy_model}
\cE^{\rm TM}(w) = \inf\left\{E^{\rm TM}(\gamma,w), \; \gamma \in \cK_2 \right\}, 
\end{equation}
where 
$$
E^{\rm TM}(\gamma,w) = \tr(H_0^{\rm TM}  \gamma)+\frac 12 \left(\tr\left((\gamma-\gamma_0^{\rm TM})^2\right)\right)^2 + \tr(\gamma w), 
$$
$$
H_0^{\rm TM} = - 2 |e_1\rangle\langle e_1| -   |e_2\rangle\langle e_2|-   |e_3\rangle\langle e_3|, \quad \gamma_0^{\rm TM}=|e_1\rangle\langle e_1|+|e_2\rangle\langle e_2|,
$$
$e_1$, $e_2$, $e_3$ being pairwise orthonormal vectors of $L^2(\R^3)$. For $w=0$, the unique ground state of (\ref{eq:toy_model}) is $\gamma_0^{\rm TM}$ and the mean-field Hamiltonian of the unperturbed system is $H_0^{\rm TM}$. We are therefore in the degenerate case with $\epsilon_{\rm F}^0=-1$ and $\delta_0^{\rm TM}=|e_2\rangle\langle e_2|$, and we have $N_{\rm f}=1$, $1 = \mbox{\rm Rank}(\delta_0^{\rm TM}) < N_{\rm p}=2$, $\mbox{\rm Ker}(1-\delta_0^{\rm TM})=\R e_2 \neq \left\{0\right\}$, so that condition~(\ref{eq:hyp_frac}) is not fulfilled. A simple calculation shows that for $w=|e_3\rangle\langle e_3|$, it holds
$$
\cE^{\rm TM}(\beta w)=\left| \begin{array}{ll} 
-3 - \frac 38 |\beta|^{4/3} &  \mbox{ for } \beta < 0, \\
-3 &  \mbox{ for } \beta \ge 0.
\end{array}\right.
$$
Clearly, real-analytic perturbation theory cannot be applied.
\end{remark}

\medskip

\begin{remark}
The  block representation of $\gamma^{(1)}_W$, the first-order term of the perturbation expansion of the ground state density matrix, is given by
\begin{equation}\label{first_perturbation}
\gamma^{(1)}_W=\left[\begin{matrix}
            0                         &   (A_{\rm pf}^{(1)})^*(1-\Lambda)  &  (A_{\rm uf}^{(1)})^* \\ \\   
           (1-\Lambda) A_{\rm pf}^{(1)}   &     A_{\rm pp}^{(1)}                   &\Lambda (A_{\rm up}^{(1)})^* \\ \\
            A_{\rm uf}^{(1)}              &  A_{\rm up}^{(1)}\Lambda           & 0          
   \end{matrix}\right],
\end{equation} 
where the above operators solve the following system
\begin{equation}\label{Euler}
 \Theta( A_{\rm uf}^{(1)},A_{\rm up}^{(1)},A_{\rm pf}^{(1)},A_{\rm pp}^{(1)})=-(W_{\rm uf},W_{\rm up}\Lambda,(1-\Lambda) W_{\rm pf},\frac{1}{2} W_{\rm pp}),
\end{equation}
where $W_{xy}$ is the $xy$-block of the operator ``multiplication by $W$''. We also have
\begin{equation*} 
 {\cal E}^{(2)}_W=\tr\left(H_0\gamma_{W,2}^{(1,1)}\right)+\frac{1}{2}D\left(\rho_{\gamma_{W,1}^{(1)}},\rho_{\gamma_{W,1}^{(1)}}\right)+\int_{\R^3}\rho_{\gamma_{W,1}^{(1)}}W.
\end{equation*}
The second-order term $\gamma^{(2)}_W$ is also useful to compute nonlinear responses. For brevity, we do not provide here the explicit formula to compute this term and refer the reader to~\cite{Mou13}.
\end{remark}

\medskip

\begin{remark} In the degenerate case, there is no analogue of (\ref{eq:CPrHF_rho}), that is no explicit closed recursion relation on the coefficients of the Rayleigh-Schr\"odinger expansion of the density.
\end{remark}

\section{Extensions to other settings}
\label{sec:extensions}

Although all the results in the preceding sections are formulated for finite molecular systems in the whole space, in the all-electron rHF framework, some of them can be easily extended to other settings:
\begin{itemize}
\item all the results in Sections 3 and 4 can be extended to valence electron calculations with nonlocal pseudopotentials, as well as to regular nonlocal perturbations of the rHF model, that is to any perturbation modeled by an operator $W$ such that $W(1-\Delta)$ is a bounded operator on $L^2(\R^3)$, the term $\int_{\R^3} \rho_\gamma W$ being then replaced with $\tr(\gamma W)$;
\item all the results in Section~3 can be extended to the rHF model for locally perturbed insulating or semiconducting crystals (see in particular~\cite{CanLew10}, where the analogues of the operators $\cL$ and $Q^{(k)}$ in Lemma~\ref{lem:linear_response} are introduced and analyzed); the extension to conducting crystals is a challenging task, see~\cite{FraLewLieSei12} for results on the particular case of the homogeneous electron gas;
\item extending our results to the Kohn-Sham LDA model for finite molecular systems in the whole space is difficult as the ground state density decays exponentially to zero at infinity while the LDA exchange-correlation energy density is not twice differentiable at $0$ (it behaves as the function $\R_+ \ni \rho \mapsto -\rho^{4/3} \in \R_-$). On the other hand, all the results in Sections~3 and 4 can be extended to the Kohn-Sham LDA model on a supercell with periodic boundary conditions as well as to the periodic Kohn-Sham LDA model for perfect crystals, as in this case, the ground state density is periodic and bounded away from zero (see e.g.~\cite{CanChaMad12,CanDelLew08}). Let us emphasize however that in the LDA setting, it is not known whether the ground state density of the unperturbed problem is unique. We must therefore restrict ourselves to local perturbation theory in the vicinity of a local minimizer and make a coercivity assumption on the Hessian of the energy functional at the unperturbed local minimizer $\gamma_0$. In the supercell setting, the operator $\cL$ was used in~\cite{ELu13} to study the stability of crystals;
\item  the Hartree-Fock model consists in minimizing the energy functional
$$
E^{\rm HF}(\gamma,W) :=    \tr\left( -\frac 12 \Delta \gamma\right) + \int_{\R^3} \rho_\gamma (V+W) + \frac 12 D(\rho_\gamma,\rho_\gamma)-\frac 12 \int_{\R^3}\int_{\R^3} \frac{|\gamma(x,y)|^2}{|x-y|} \, dx \; dy
$$
over the set $\cP_N$ of Slater determinants with finite kinetic energy. It turns out that all the local minimizers of $E^{\rm HF}(\gamma,W)$ on $\cK_N$ are on $\cP_N$ (Lieb's variational principle~\cite{Lie77}). Consequently, an equivalent formulation of the Hartree-Fock model is
\begin{equation}\label{eq:HF}
\cE(W):=\inf\left\{ E^{\rm HF}(\gamma,W), \; \gamma \in \cK_N \right\}.
\end{equation}
Uniqueness for problem (\ref{eq:HF}) is an essentially open question (see however~\cite{Gri12} for partial results). In order to apply perturbation theory, we therefore need a coercivity assumption on the Hessian at the minimizer $\gamma_0$, just as in the LDA setting. It is known that there are no unfilled shells in the Hartree-Fock theory \cite{BacLieLosSol94}, which implies that we are always in the non-degenerate case. The first three statements and the fifth statement of Theorem~\ref{th:CPrHF} can be transposed to the Hartree-Fock setting under the above mentioned coercivity assumption. On the other hand, there is no analogue of (\ref{eq:CPrHF_rho}) for the Hartree-Fock model. A mathematical analysis of the perturbation theory for the molecular orbital formulation of the Hartree-Fock model was published in~\cite{CanLeB98}. It is easily checked that our proof of Wigner's $(2n+1)$-rule also applies to the Hartree-Fock setting;
\item the extension to some of our results to Stark potentials $W(x) = - E \cdot x$, where $E \in \R^3$ is a uniform electric field, will be dealt with in a future work~\cite{CanMou14}. 
\end{itemize}

\section{Proofs}
\label{sec:proofs}

\subsection{Proof of Lemma~\ref{lem:uniqueness}}

Let $\gamma_0$ and $\gamma_0'$ be two ground states of (\ref{eq:min_rHF}) for $W=0$. By Theorem~\ref{th:GS_rHF}, $\gamma_0-\gamma'_0= \sigma$, with $\sigma \in \cS(L^2(\R^3))$, $\mbox{Ran}(\sigma) \subset \mbox{Ker}(H_0-\epsilon_{\rm F}^0)$, $\tr(\sigma)=0$. Therefore,
$$
\sigma = \sum_{i,j=1}^{N_{\rm p}} M_{ij} |\phi_{N_{\rm f}+i}^0\rangle\langle \phi_{N_{\rm f}+j}^0|
$$
for some symmetric matrix $M \in \R^{N_{\rm p}\times N_{\rm p}}_{\rm S}$ such that $\tr(M)=0$. As, still by Theorem~\ref{th:GS_rHF}, $\gamma_0$ and $\gamma'_0$ share the same density, the density of $\sigma$ is identically equal to zero, that is
$$
\forall x \in \R^3, \quad  \sum_{i,j=1}^{N_{\rm p}} M_{ij}\phi_{N_{\rm f}+i}^0(x)\phi_{N_{\rm f}+j}^0(x)  = 0.
$$
If Assumption (\ref{eq:hyp_ND}) is satisfied, then $M=0$; therefore $\sigma=0$, and uniqueness is proved.

\subsection{Proof of Proposition~\ref{prop:uniqueness}}

Let us first notice that as for all $1 \le i \le N_{\rm p}$,  $\phi_{N_{\rm f}+i}^0 \in D(H_0)=H^2(\R^3) \hookrightarrow C^0(\R^3)$, condition (\ref{eq:hyp_ND}) is mathematically well-defined.  

\medskip

\noindent
{\bf Case 1:} Let $M \in \R^{N_{\rm p} \times N_{\rm p}}_{\rm S}$ be such that 
$$
\forall x \in \R^3, \quad \sum_{i,j=1}^{N_{\rm p}} M_{ij}\phi_{N_{\rm f}+i}^0(x)\phi_{N_{\rm f}+j}^0(x) = 0.
$$
The matrix $M$ being symmetric, there exists an orthogonal matrix $U \in O(N_{\rm p})$ such that $UMU^T = \mbox{diag}(n_1,\cdots,n_{N_{\rm p}})$ with $n_1 \le \cdots \le n_{N_{\rm p}}$. Let $\widetilde\phi_{N_{\rm f}+i}^0(x) = \sum_{j=1}^{N_{\rm p}} U_{ij} \phi_{N_{\rm f}+j}^0(x)$. The functions $\widetilde\phi_{N_{\rm f}+i}^0$ form an orthonormal basis of $\mbox{Ker}(H_0-\epsilon_{\rm F}^0)$ and it holds
$$
\forall x \in \R^3, \quad \sum_{i=1}^{N_{\rm p}} n_i |\widetilde \phi_{N_{\rm f}+i}^0(x)|^2 = 0,
$$
from which we deduce that $\sum_{i=1}^{N_{\rm p}}n_i=0$. Consider first the case when $N_{\rm p}=2$. If $M \neq 0$, then $n_2=-n_1=n > 0$, so that 
$$
\forall x \in \R^3, \quad |\widetilde \phi_{N_{\rm f}+1}^0(x)|^2 = |\widetilde \phi_{N_{\rm f}+2}^0(x)|^2.
$$
In particular, the two eigenfunctions $\widetilde \phi_{N_{\rm f}+1}^0$ and $\widetilde \phi_{N_{\rm f}+2}^0$ have the same nodal surfaces (that is $(\widetilde \phi_{N_{\rm f}+1}^0)^{-1}(0)=(\widetilde \phi_{N_{\rm f}+2}^0)^{-1}(0)$). Consider now the case when $N_{\rm p}=3$. If $M \neq 0$, then either $n_2=0$ and $\widetilde \phi_{N_{\rm f}+1}^0$ and $\widetilde \phi_{N_{\rm f}+3}^0$ have the same nodes, or $n_2 \neq 0$. Replacing $M$ with $-M$, we can, without loss of generality assume that $n_1 < 0 < n_2 \le n_3$, which leads to 
$$
\forall x \in \R^3, \quad |\widetilde \phi_{N_{\rm f}+1}^0(x)|^2 = \frac{|n_2|}{|n_1|} |\widetilde \phi_{N_{\rm f}+2}^0(x)|^2 + \frac{|n_3|}{|n_1|} |\widetilde \phi_{N_{\rm f}+3}^0(x)|^2.
$$
We infer from the above equality that the nodal surfaces of $\widetilde \phi_{N_{\rm f}+1}^0(x)$ are included in those of $\widetilde \phi_{N_{\rm f}+2}^0(x)$. Let $\varOmega$ be a connected component of the open set $\R^{3} \setminus (\widetilde \phi_{N_{\rm f}+1}^0)^{-1}(0)$, and let $H_0^\varOmega$ be the self-adjoint operator on $L^2(\varOmega)$ with domain
$$
D(H_0^\varOmega)=\left\{ u \in H^1_0(\varOmega) \; | \; \Delta u \in L^2(\varOmega)\right\}
$$
defined by
$$
\forall u \in D(H_0^\varOmega), \quad H_0^\varOmega u = -\frac 12 \Delta u + Vu + (\rho_0\star |\cdot|^{-1}) u.
$$
As both $\psi_1=\widetilde \phi_{N_{\rm f}+1}^0|_\varOmega$ and $\psi_2=\widetilde \phi_{N_{\rm f}+2}^0|_\varOmega$ are in $D(H_0^\varOmega)$ and satisfy $H_0^\varOmega \psi_1=\epsilon_{\rm F}^0\psi_1$, $H_0^\varOmega \psi_2=\epsilon_{\rm F}^0\psi_2$, $|\psi_1| > 0$ in $\varOmega$, we deduce from \cite[Theorem XIII.44]{ReeSim78} that $\epsilon_{\rm F}^0$ is the non-degenerate ground state eigenvalue of $H_0^\varOmega$, so that there exists a real constant $C \in \R$ such that $\psi_2=C\psi_1$. It follows from the unique continuation principle (see e.g.~\cite[Theorem~XIII.57]{ReeSim78}) that $\widetilde \phi_{N_{\rm f}+2}^0=C \widetilde \phi_{N_{\rm f}+1}^0$ on $\R^3$, which contradicts the fact that $\widetilde \phi_{N_{\rm f}+1}^0$ and $\widetilde \phi_{N_{\rm f}+2}^0$ are orthogonal and non identically equal to zero. Thus, $M=0$ and the proof of case~1 is complete. 

\medskip

\noindent
{\bf Case 2.}  The degeneracy being assumed essential, $\epsilon_{\rm F}^0$ is $(2l+1)$-times degenerate for some integer $l \ge 1$, and there exists an orthonormal basis of associated eigenfunctions of the form
\begin{equation*}
\forall 1 \le i \le N_{\rm p}=2l+1, \quad \phi_{N_{\rm f}+i}^0(x)=R_l(r) \, {\cal Y}_l^{-l+i-1}(\theta,\varphi),
\end{equation*}
where $(r,\theta,\varphi)$ are the spherical coordinates of the point $x \in \R^3$, and where the functions ${\cal Y}_l^m$ are the spherical harmonics. In particular,
\begin{equation*}
 \sum_{i,j=1}^{2l+1}M_{ij}\phi_{N_{\rm f}+i}^0(x)\phi_{N_{\rm f}+j}^0(x)=R_l(r)^2\sum_{i,j=1}^{2l+1}M_{ij} \, {\cal Y}_l^{-l+i-1}(\theta,\varphi) \, {\cal Y}_l^{-l+j-1}(\theta,\varphi).
\end{equation*}
We therefore have to prove that for any symmetric matrix $M \in \R_{\rm S}^{(2l+1)\times(2l+1)}$,
$$
\left( \sum_{i,j=1}^{2l+1}M_{ij}{\cal Y}_l^{-l+i-1}{\cal Y}_l^{-l+j-1} = 0 \right) \; \Rightarrow \; M=0.
$$
Let $M \in \R_{\rm S}^{(2l+1)\times(2l+1)}$ a symmetric matrix such that 
$$
\sum_{i,j=1}^{2l+1}M_{ij}{\cal Y}_l^{-l+i-1}{\cal Y}_l^{-l+j-1} = 0
$$
on the unit sphere ${\mathbb S}^2$. Using the relation
$$
 {\cal Y}_{l}^{m_1}{\cal Y}_{l}^{m_2}= \sum_{L=0}^{2l}\sqrt{\dfrac{(2l+1)^2(2L+1)}{4\pi}} 
\left(\begin{matrix} l &l &L \\ \\ m_1 &m_2 & -(m_1+m_2)\end{matrix}\right)     \left(\begin{matrix} l & l &L \\ \\ 0 & 0 & 0\end{matrix}\right) {\cal Y}_L^{m_1+m_2},
$$
where the $\left( \begin{matrix} l_1 &l_2 & l_3 \\ \\ m_1 &m_2 & m_3 \end{matrix}\right)$ denote the Wigner 3-j symbols (see~\cite{BriSat93} for instance), and where, by convention, ${\cal Y}_L^m=0$ whenever $|m| > L$, we obtain 
\begin{eqnarray*}
0&=&\dps \frac{\sqrt{4\pi}}{2l+1}\sum_{i,j=1}^{2l+1}M_{ij}{\cal Y}_l^{-l+i-1}{\cal Y}_l^{-l+j-1} \nonumber \\ \nonumber  \\
&=&\dps \sum_{i,j=1}^{2l+1}M_{ij}\sum_{L=0}^{2l}  \sqrt{2L+1}
\left(\begin{matrix} l &l &L \\ \\ -l+i-1 & -l+j-1  & 2l+2-i-j \end{matrix}\right)\left(\begin{matrix} l & l &L \\ \\ 0 & 0 & 0\end{matrix}\right) {\cal Y}_L^{i+j-2l-2} \nonumber \\ \nonumber  \\
&=&\dps \sum_{m=-2l}^{2l}\sum_{L=0}^{2l} \sqrt{2L+1} \left(\begin{matrix} l & l &L \\ \\ 0 & 0 & 0\end{matrix}\right) \left[\sum_{\begin{array}{c} \mbox{\small{$1 \le i,j \le 2l+1$}} \\  \mbox{\small{$i+j-2l-2=m$}} \end{array}} 
\left(\begin{matrix} l &l &L \\ \\ -l+i-1  & -l+j-1  & -m\end{matrix}\right)M_{ij} \right]{\cal Y}_L^{m}. 
\end{eqnarray*} 
Using the fact that the Wigner 3-j symbol $\left( \begin{matrix} l & l & L \\ \\ m_1 &m_2 & -(m_1+m_2) \end{matrix}\right)$ is equal to zero unless
$$
|m_1| \le l, \quad |m_2| \le l, \quad |m_1+m_2| \le L, \quad 0 \le L \le 2l, \quad {\rm and} \quad L \in 2\N \mbox{ if } m_1=m_2=0,
$$
we obtain that for all $L \in \left\{0,2,\cdots,2l \right\}$ and all $-L \le m \le L$,
\begin{equation}\label{eq:linearsystem}
\sum_{\begin{array}{c} \mbox{\small{$1 \le i,j \le 2l+1$}} \\  \mbox{\small{$i+j-2l-2=m$}} \end{array}} 
\left(\begin{matrix} l &l &L \\ \\ -l+i-1  & -l+j-1  & -m\end{matrix}\right)M_{ij}  = 0.
\end{equation}
For $m=-2l$ and $L=2l$, the above expression reduces to 
$$
\left(\begin{matrix} l &l &2l \\ \\ -l  & -l  & 2l\end{matrix}\right)M_{11}  = 0, \qquad \mbox{where} \qquad \left(\begin{matrix} l &l &2l \\ \\ -l  & -l  & 2l\end{matrix}\right)= \frac{1}{\sqrt{4l+1}}.
$$
Hence $M_{11}=0$. More generally, for each integer value of $m$ in the range $[-2l,2l]$, equation (\ref{eq:linearsystem}) gives rise to a linear system of $n_{m,l}$ equations (obtained for the various even values of $L$ in the range $[|m|,2l]$) with $n_{m,l}$ unknowns (the $M_{i,j}$'s satisfying $i \le j$ - recall that the matrix $M$ is symmetric - and $i+j=2l+2+m$).
Using the symmetry property
$$
 \left(\begin{matrix} l &l & L \\ \\ -l+i-1 & -l+j-1 & -m\end{matrix}\right) =  \left(\begin{matrix} l &l & L \\ \\ -l+j-1 & -l+i-1 & -m\end{matrix}\right)
$$ 
and the orthogonality relation stating that for all $-2l \le m \le 2l$, and all $|m| \le L,L' \le 2l$, 
$$
\displaystyle \sum_{\begin{array}{c} \mbox{\small{$1 \le i,j \le 2l+1$}} \\  \mbox{\small{$i+j-2l-2=m$}} \end{array}} \left(\begin{matrix} l &l & L \\ \\ -l+i-1 & -l+j-1 & -m\end{matrix}\right)
                                 \left(\begin{matrix} l & l & L' \\ \\ -l+i-1 & -l+j-1 & -m \end{matrix}\right)=\frac{\delta_{LL'}}{(2L+1)},
$$
it is easy to see that this linear system is free, and that the corresponding entries of $M$ are therefore equal to $0$. Hence, the matrix $M$ is identically equal to zero, which completes the proof.

\subsection{Proof of Lemma~\ref{lem:linear_response}}

As $\gC$ is a compact subset of the resolvent set of $H_0$ and as the domain of $H_0$ is $H^2(\R^3)$, there exists $C_0 \in \R_+$ such that
$$
\max_{z \in \gC}(\|(z-H_0)^{-1}\|,\|(1-\Delta)(z-H_0)^{-1}\|, \|(z-H_0)(1-\Delta)^{-1}\|) \le C_0.
$$
It follows from the Kato-Seiler-Simon inequality~\cite{Sim79} that for all $v \in \cC'$, 
$$
\| v(z-H_0)^{-1}\| \le C_0 \|v (1-\Delta)^{-1} \| \le C_0 \| v(1-\Delta)^{-1} \|_{\gS_6} \le C \|v\|_{L^6} \le  \alpha \|v\|_{\cC'},
$$
for constants $\alpha,C \in \R_+$ independent of $v$. The $k$-linear map $Q^{(k)}$ is therefore well-defined and continuous from $(\cC')^k$ to the space of bounded operators on $L^2(\R^3)$. Denoting by $\gamma_0^\perp=1-\gamma_0$, we have 
$$
Q^{(k)}(v_1,\cdots,v_k)= \sum_{(P_j)_{0 \le j \le k}\in \left\{\gamma_0,\gamma_0^\perp\right\}^{k+1}} \frac{1}{2i\pi}\oint_\gC  (z-H_0)^{-1} P_0 \prod_{j=1}^k \left( v_j(z-H_0)^{-1} P_{j} \right) \, dz.
$$
In the above sum, the term with all the $P_j$'s equal to $\gamma_0^\perp$ is equal to zero as a consequence of Cauchy's residue formula. In all the remaining terms, one of the $P_j$'s is equal to the rank-$N$ operator $\gamma_0$. The operators $(z-H_0)^{-1}$ and $v_j(z-H_0)^{-1}$ being bounded, $Q^{(k)}(v_1,\cdots,v_k)$ is finite-rank, hence trace-class, and it holds
$$
\|Q^{(k)}(v_1,\cdots,v_k)\|_{\gS_1} \le \frac{|\gC|}{2\pi} NC_0 \alpha^{k} \|v_1\|_{\cC'} \cdots  \|v_k\|_{\cC'}.
$$
Likewise, the operator 
\begin{eqnarray*}
  \!\!\!\!\!\!\!\!\!\!\!\!&& |\nabla|Q^{(k)}(v_1,\cdots,v_k)|\nabla| \\
  \!\!\!\!\!\!\!\!\!\!\!\!\!\!\!\!&&=  \!\!\!\!\!\!\!\! \sum_{(P_j)\in \left\{\gamma_0,\gamma_0^\perp\right\}^{k+1}} \frac{1}{2i\pi}\oint_\gC  |\nabla|(z-H_0)^{-1/2} P_0 \prod_{j=1}^k \left((z-H_0)^{-1/2} v_j(z-H_0)^{-1/2} P_{j} \right) (z-H_0)^{-1/2}|\nabla| \, dz
\end{eqnarray*}
is finite rank and 
$$
\|\, |\nabla|Q^{(k)}(v_1,\cdots,v_k)|\nabla|\, \|_{\gS_1} \le C \alpha^{k} \|v_1\|_{\cC'} \cdots  \|v_k\|_{\cC'},
$$
for some constant $C$ independent of $v_1,\cdots,v_k$. Therefore $Q^{(k)}$ is a continuous linear map from $(\cC')^k$ to $\gS_{1,1}$ and the bound (\ref{eq:bounds_Qk}) holds true. It then follows from Cauchy's residue formula and the cyclicity of the trace that, for $k\geq 1$,
\begin{eqnarray*}
\!\!\!\!&&\!\!\!\!\!\!\!\!\tr(Q^{(k)}(v_1,\cdots,v_k))= \tr \left( \frac{1}{2i\pi}\oint_\gC (z-H_0)^{-1}  \prod_{j=1}^k \left( v_j(z-H_0)^{-1} \right)    \, dz \right) \\
&&= \sum_{(P_j)\in \left\{\gamma_0,\gamma_0^\perp\right\}^{k+1}}  \tr \left( \frac{1}{2i\pi}\oint_\gC (z-H_0)^{-1} P_0 \prod_{j=1}^k \left( v_j(z-H_0)^{-1} P_{j} \right)     \, dz \right) \\
&&= \sum_{j=1}^k \sum_{(P_l)\in \left\{\gamma_0,\gamma_0^\perp\right\}^{k}}  \tr \left( \frac{1}{2i\pi}\oint_\gC  \prod_{l=1}^{k-1} \left( v_{l+j \; {\rm mod}(k)} (z-H_0)^{-1} P_{l} \right)   v_j   (z-H_0)^{-2} \gamma_0  \, dz \right) =0.
\end{eqnarray*}

Let $\rho \in \cC$ and $Q:=Q^{(1)}(\rho \star |\cdot|^{-1})$. Proceeding as above, we obtain that for all $\phi \in C^\infty_{\rm c}(\R^3)$, 
\begin{eqnarray*}
\left|\int_{\R^3} \rho_{Q}\phi \right| &=& \left| \tr\left( Q \phi \right) \right| = \left| \tr\left( \frac 1{2i\pi} \oint_\gC  (z-H_0)^{-1}(\rho \star |\cdot|^{-1})(z-H_0)^{-1}\phi  \, dz \right) \right| \\
&\le & C \|\rho\|_{\cC} \|\phi\|_{\cC'},
\end{eqnarray*}
for a constant $C\in \R_+$ independent of $\rho$ and $\phi$. Therefore, $\rho_Q$ is in $\cC$ and $\|\rho_Q\|_{\cC} \le  C \|\rho\|_{\cC}$. This proves that $\cL$ is a bounded operator on $\cC$. In addition, for all $\rho_1,\rho_2$ in $\cC$,
$$
(\cL\rho_1,\rho_2)_\cC=  - \tr\left(  \frac 1{2i\pi} \oint_\gC (z-H_0)^{-1}(\rho_1 \star |\cdot|^{-1})(z-H_0)^{-1}(\rho_2 \star |\cdot|^{-1}) \, dz \right) = (\rho_1,\cL\rho_2)_\cC,
$$
where we have used again the cyclicity of the trace. Thus, $\cL$ is self-adjoint. Lastly, for all $\rho \in \cC$,
$$
(\cL\rho,\rho)_\cC= \sum_{i=1}^N \langle \gamma_0^\perp ((\rho\star|\cdot|^{-1})\phi_i^0) | (H_0^\perp-\epsilon_i)^{-1} |  \gamma_0^\perp ( (\rho\star|\cdot|^{-1})\phi_i^0)   \rangle \ge 0,
$$
where $H_0^{\perp}$ is the self-adjoint operator on $\mbox{Ran}(\gamma_0^\perp)=\mbox{Ker}(\gamma_0)$ defined by $\forall v \in \mbox{Ran}(\gamma_0^\perp)$, $H_0^\perp v =H_0 v$.

\subsection{Stability of the spectrum of the mean-field Hamiltonian}

We assume here that we are 
\begin{itemize}
\item either in the non-degenerate case ($\epsilon_N < 0$ and $\epsilon_N < \epsilon_{N+1}$), in which case we set $\epsilon_{\rm F}^0=\frac{\epsilon_N+\epsilon_{N+1}}2$;
\item or in the degenerate case ($\epsilon_N=\epsilon_{N+1}=\epsilon_{\rm F}^0<0$).
\end{itemize}
We recall that $N_{\rm f}=\mbox{Rank}(\1_{(-\infty,\epsilon_{\rm F}^0)}(H_0))$, $N_{\rm p}=\mbox{Rank}(\1_{\left\{\epsilon_{\rm F}^0\right\}}(H_0))$ and $N_{\rm o}=N_{\rm f}+N_{\rm p}$. We also have $g_-=\epsilon_{\rm F}^0-\epsilon_{N_{\rm f}}$ and $g_+=\epsilon_{N_{\rm f}+N_{\rm p}+1}-\epsilon_{\rm F}^0$. By definition $g_- > 0$ and $g_+ > 0$ since $\epsilon_{\rm F}^0 < 0$. 

\medskip

\begin{lemma}\label{lem:stability} Let 
$$
\alpha_1=\epsilon_1-1, \;  \alpha_2=\epsilon_{\rm F}^0-\frac{3g_-}4,   \; \alpha_3=\epsilon_{\rm F}^0-\frac{g_-}4,   \;  \alpha_4=\epsilon_{\rm F}^0+\frac{g_+}4,   \; \alpha_5=\epsilon_{\rm F}^0+\frac{3g_+}4.  
$$
There exists $\eta > 0$ such that for all $v \in B_\eta(\cC')$, 
\begin{eqnarray*}
&& \!\!\!\!\!\!\!\!\!\!\!\! \mbox{\rm Rank}(\1_{(-\infty,\alpha_1]}(H_0+v))=0, \; \mbox{\rm Rank}(\1_{(\alpha_1,\alpha_2)}(H_0+v))=N_{\rm f}, \; \mbox{\rm Rank}(\1_{[\alpha_2,\alpha_3]}(H_0+v))=0, \\
&& \!\!\!\!\!\!\!\!\!\!\!\! \mbox{\rm Rank}(\1_{(\alpha_3,\alpha_4]}(H_0+v))=N_{\rm p}, \; \mbox{\rm Rank}(\1_{(\alpha_4,\alpha_5]}(H_0+v))=0. 
\end{eqnarray*}
\end{lemma}

\medskip

\begin{proof} Let $z \in \left\{\alpha_1,\alpha_2,\alpha_3,\alpha_4,\alpha_5\right\}$. As $z \notin \sigma(H_0)$, we have
$$
z-(H_0+v) =   \left(1+v (1-\Delta)^{-1} (1-\Delta)(z-H_0)^{-1} \right) \, (z-H_0).
$$
Besides, as $D(H_0)=H^2(\R^3)$, there exists a constant $C \in \R_+$ independent of the choice of $z\in \left\{\alpha_1,\alpha_2,\alpha_3,\alpha_4,\alpha_5\right\}$, such that
$$
\| (1-\Delta) (z-H_0)^{-1} \| \le C.
$$
In addition, there exists a constant $C' \in \R_+$ such that for all $v \in \cC'$, 
$$
\|v(1-\Delta)^{-1}\| \le \| v(1-\Delta)^{-1} \|_{\gS_6} \le  C' \|v\|_{\cC'}.
$$
Let $\eta=(CC')^{-1}$. We obtain that for all $v \in B_\eta(\cC')$, 
$$
\|v (1-\Delta)^{-1} (1-\Delta)(z-H_0)^{-1}\| < 1,
$$
so that $z-(H_0+v)$ is invertible. Therefore, for all $v \in B_\eta(\cC')$, none of the real numbers $\alpha_1,\alpha_2,\alpha_3,\alpha_4,\alpha_5$ are in $\sigma(H_0+v)$. It also follows from the above arguments that for all $v \in \cC'$, the multiplication by $v$ is a $H_0$-bounded operator on $L^2(\R^3)$. Using Kato's perturbation theory, we deduce from a simple continuity argument that the ranks of the spectral projectors 
$$
\1_{(-\infty,\alpha_1]}(H_0+v), \; \1_{(\alpha_1,\alpha_2)}(H_0+v), \; \1_{[\alpha_2,\alpha_3]}(H_0+v), \; \1_{(\alpha_3,\alpha_4]}(H_0+v), \mbox{ and } \1_{(\alpha_4,\alpha_5]}(H_0+v)
$$  
are constant for $v \in B_\eta(\cC')$, and therefore equal to their values for $v=0$, namely $0$, $N_{\rm f}$, $0$, $N_{\rm p}$ and $0$ respectively.
\end{proof}

\subsection{Proof of Theorem~\ref{th:CPrHF}}

\noindent
{\bf Step 1:} proof of statement 1. 

\medskip

\noindent
Let us introduce the relaxed constrained problem
\begin{equation} \label{eq:relaxed}
{\cE}^{\rm rHF}_{\le N}(W) =  \inf \left\{E^{\rm rHF}(\gamma,W), \; \gamma \in \cK_{\le N} \right\},
\end{equation}
where
$$
\cK_{\le N}=\left\{ \gamma \in \cS(L^2(\R^3)) \; | \; 0 \le \gamma \le 1, \; \tr(\gamma) \le N, \; \tr(-\Delta\gamma) < \infty \right\}.
$$
As $\epsilon_{\rm F}^0 < 0$, $\gamma_0$ is the unique minimizer of (\ref{eq:relaxed}) for $W=0$, and as $\cK_{\le N}$ is convex, the corresponding optimality condition reads
\begin{equation}\label{eq:EulerKleN}
\forall \gamma \in \cK_{\le N}, \quad \tr(H_0(\gamma-\gamma_0)) \ge 0.
\end{equation}

Let $W \in \cC'$, and $(\gamma'_k)_{k \in \N^\ast}$ a minimizing sequence for (\ref{eq:relaxed}) for which
\begin{equation}\label{eq:minseq2}
\forall k \ge 1, \quad E^{\rm rHF}(\gamma'_k,W) \le \cE^{\rm rHF}_{\le N}(W)+\frac 1 k.
\end{equation}
Set $\rho'_k=\rho_{\gamma'_k}$. We obtain on the one hand, using (\ref{eq:EulerKleN}),
\begin{eqnarray*}
\cE^{\rm rHF}_{\le N}(W) &\ge & E^{\rm rHF}(\gamma'_k,W) - \frac 1 k \\
&=& E^{\rm rHF}(\gamma'_k,0) + \int_{\R^3} \rho'_k W  - \frac 1 k \\
&=&  \cE^{\rm rHF}_{\le N}(0) + \tr(H_0(\gamma'_k-\gamma_0)) + \frac 12 D(\rho'_k-\rho_0, \rho'_k-\rho_0) + \int_{\R^3} \rho'_k W  - \frac 1 k \\
& \ge & \cE^{\rm rHF}_{\le N}(0) +\frac 12 D(\rho'_k-\rho_0, \rho'_k-\rho_0) + \int_{\R^3} \rho'_k W - \frac 1 k ,
\end{eqnarray*}
and on the other hand
$$
\cE^{\rm rHF}_{\le N}(W) \le E^{\rm rHF}(\gamma_0,W) = \cE^{\rm rHF}_{\le N}(0) + \int_{\R^3} \rho_0 W  .
$$
Therefore,
$$
 \frac 12 D(\rho'_k-\rho_0, \rho'_k-\rho_0) \le \int_{\R^3} (\rho_0-\rho'_k) W +\frac 1 k,
$$
from which we get 
$$
\frac 12 \| \rho'_k-\rho_0 \|_{\cC}^2 \le \|W\|_{\cC'} \| \rho'_k-\rho_0 \|_{\cC} + \frac 1k,
$$
and finally
\begin{equation}\label{eq:boundrhok}
\|\rho'_k-\rho_0 \|_{\cC} \le 2 \|W\|_{\cC'} + \left( 2k^{-1} \right)^{1/2}.
\end{equation}
Then, using Cauchy-Schwarz, Hardy and Hoffmann-Ostenhof~\cite{Hof77} inequalities, we obtain
\begin{eqnarray*}
\cE^{\rm rHF}_{\le N}(0) &=& \cE^{\rm rHF}(0)= E^{\rm rHF}(\gamma_0,0) = E^{\rm rHF}(\gamma_0,W) -\int_{\R^3} \rho_0 W \\
&\ge& \cE^{\rm rHF}_{\le N}(W) -\int_{\R^3} \rho_0 W
\ge  E^{\rm rHF}(\gamma'_k,W) -\int_{\R^3} \rho_0 W -\frac 1 k \\ 
  &=& \frac 12 \tr(-\Delta\gamma'_k) + \int_{\R^3} V \rho'_k + \frac 12 D(\rho'_k,\rho'_k) + \int_{\R^3} \rho'_kW-\int_{\R^3} \rho_0 W - \frac 1 k\\
& \ge & \frac 12 \tr(-\Delta\gamma'_k) - 2 Z N^\frac{1}{2} (\tr(-\Delta \gamma'_k))^{1/2} + \frac 12 \| \rho'_k \|_{\cC}^2 -   \| \rho'_k \|_{\cC} \|W\|_{\cC'}-  \| \rho_0 \|_{\cC} \|W\|_{\cC'} - \frac 1 k\\
& \ge &  \frac 12 ((\tr(-\Delta\gamma'_k))^{1/2}-2ZN^\frac 12)^2 + \frac 12 (\| \rho'_k \|_{\cC}-\|W\|_{\cC'})^2 - 2Z^2N-\frac 12 \| \rho_0 \|_{\cC}^2- \|W\|_{\cC'}^2- \frac 1 k \\
& \ge &  \frac 12 ((\tr(-\Delta\gamma'_k))^{1/2}-2ZN^\frac 12)^2 - 2Z^2N-\frac 12 \| \rho_0 \|_{\cC}^2- \|W\|_{\cC'}^2- \frac 1 k,
\end{eqnarray*}
from which we infer that 
$$
\tr(-\Delta \gamma'_k) \le C_0 (1+\|W\|_{\cC'}^2),
$$
for some constant $C_0 \in \R_+$ independent of $W$ and $k$. This estimate, together with (\ref{eq:boundrhok}) and the fact that $\|\gamma'_k\|_{\gS_1} = \tr(\gamma'_k) \le N$, shows that the sequences $(\gamma'_{k})_{k \in \N^\ast}$ and $(\rho'_{k})_{k \in \N^\ast}$ are bounded in $\gS_{1,1}$ and $\cC$ respectively. We can therefore extract from $(\gamma'_k)_{k \in \N^\ast}$ a subsequence $(\gamma'_{k_j})_{j \in \N^\ast}$ such that $(\gamma'_{k_j})_{j \in \N}$ converges to $\gamma_W$ for the weak-$*$ topology of $\gS_{1,1}$, and $(\rho'_{k_j})_{j \in \N}$ converges to $\rho_W:=\rho_{\gamma_W}$ weakly in $\cC$ and strongly in $L^p_{\rm loc}(\R^3)$ for all $1 \le p < 3$. This implies that 
$$
\gamma_W \in \cK_{\le N} \quad \mbox{and} \quad E^{\rm rHF}(\gamma_W,W) \le \liminf_{j \to \infty} E^{\rm rHF}(\gamma'_{k_j},W) = \cE^{\rm rHF}_{\le N}(W).
$$ 
Thus $\gamma_W$ is a minimizer of (\ref{eq:relaxed}). In addition, as the rHF model is strictly convex in the density, all the minimizers of (\ref{eq:relaxed}) have the same density $\rho_W$, and, passing in the limit in (\ref{eq:boundrhok}), we obtain that $\rho_W$ satisfies 
$$
\|\rho_W-\rho_0 \|_{\cC} \le 2 \|W\|_{\cC'}.
$$
Denoting by 
\begin{equation}\label{eq:Veff}
v_W=W +(\rho_W-\rho_0)\star |\cdot|^{-1},
\end{equation}
we have
\begin{equation}\label{eq:HW2}
H_W = -\frac 12 \Delta + V + W + \rho_W \star |\cdot|^{-1} 
= H_0 + v_W,
\end{equation}
with 
\begin{equation}\label{eq:norm_vW}
\|v_W\|_{\cC'} \le \|W\|_{\cC'} + \|(\rho_W -\rho_0)\star |\cdot|^{-1}\|_{\cC'} \le 3\|W\|_{\cC'}.
\end{equation} 
By Lemma~\ref{lem:stability}, for all $W \in B_{\eta/3}(\cC')$, we have 
$$
\mbox{Rank}(\1_{(-\infty,\epsilon_{\rm F}^0-g_-/2]}(H_W))= N \quad \mbox{and} \quad \mbox{Rank}(\1_{(\epsilon_{\rm F}^0-g_-/2,\epsilon_{\rm F}^0+g_-/2]}(H_W))= 0. 
$$
In particular, $H_W$ has a least $N$ negative eigenvalues, from which we infer that $\tr(\gamma_W)=N$. Therefore, $\gamma_W$ is a minimizer of (\ref{eq:min_rHF}). In addition, $\gamma_W=\1_{(-\infty,\epsilon_{\rm F}^0]}(H_W)$ and it holds
\begin{equation}\label{eq:CauchyGammaW}
\gamma_W = \frac{1}{2i\pi}\oint_\gC (z-H_W)^{-1} \, dz.
\end{equation}

\medskip

\noindent
{\bf Step~2:} proof of statement~2.

\medskip

\noindent
It follows from  (\ref{eq:Veff}),  (\ref{eq:HW2}) and (\ref{eq:CauchyGammaW}) that 
$$
\forall W \in B_{\eta/3}(\cC'), \quad {\cal X}(v_W)=W,
$$
where ${\cal X}$ is the mapping from $B_{\eta/3}({\cal C}')$ to ${\cal C}'$ defined by 
$$
{\cal X}(v) = v - \rho_{\frac{1}{2i\pi} \oint_\gC ((z-H_0-v)^{-1}-(z-H_0)^{-1}) \, dz} \star |\cdot|^{-1}.
$$
The mapping ${\cal X}$ is real analytic. Besides, denoting by $v_{\rm c}$ the Coulomb operator associating to each density $\rho \in \cC$ the electrostatic potential $v_{\rm c}(\rho) = \rho \star |\cdot|^{-1} \in \cC'$, we have
$$
{\cal X}'(0) = v_{\rm c} (1+{\cal L}) v_{\rm c}^{-1}.
$$
It follows from the second statement of Lemma~\ref{lem:linear_response} and from the fact that $v_{\rm c}:\cC \rightarrow \cC'$ is a bijective isometry that ${\cal X}'(0)$ is bijective. Applying the real analytic implicit function theorem, we obtain that the mapping $W \mapsto v_W$ is real analytic from some ball $B_{\eta'}(\cC')$ (for some $\eta' > 0$) to $\cC'$. By composition of real analytic functions, the functions 
$$
\gamma_W = \frac{1}{2i\pi} \oint_\gC (z-H_0-v_W)^{-1} \, dz, \; \rho_W=\rho_0+v_{\rm c}^{-1}(v_W-W) \mbox{ and } \cE^{\rm rHF}(W) = E^{\rm rHF}(\gamma_W,W)
$$
are real analytic from $B_{\eta'}(\cC')$ to $\gS_{11}$, $\cC$ and $\R$ respectively.

\medskip

\noindent
{\bf Step~3:} proof of statements~3 and 4.

\medskip

\noindent
Let $W \in B_{\eta'}(\cC')$. It follows from the above result that the functions $\beta \mapsto \gamma_{\beta W}$,  $\beta \mapsto \rho_{\beta W}$, and $\beta \mapsto \cE^{\rm rHF}(\beta W)$ are real analytic in the vicinity of $0$, so that, for $|\beta|$ small enough,
$$
\gamma_{\beta W} = \gamma_0 + \sum_{k=1}^{+\infty} \beta^k \gamma_W^{(k)}, \quad\rho_{\beta W} = \rho_0 + \sum_{k=1}^{+\infty} \beta^k \rho_W^{(k)}, \quad \cE^{\rm rHF}(\beta W) = \cE^{\rm rHF}(0) + \sum_{k=1}^{+\infty} \beta^k \cE_W^{(k)}, 
$$
the series being normally convergent in $\gS_{11}$, $\cC$ and $\R$ respectively. The Dyson expansion of (\ref{eq:gammaWint}) gives
\begin{eqnarray*}
\gamma_{\beta W} &=& \gamma_0 + \sum_{k=1}^{+\infty} Q^{(k)}\left(v_{\beta W}, \cdots , v_{\beta W} \right).
\end{eqnarray*}
As
$$
v_{\beta W}
= \beta W + \sum_{k=1}^{+\infty} \beta^k (\rho_W^{(k)} \star |\cdot|^{-1})
= \sum_{k=1}^{+\infty} \beta^k W^{(k)},
$$
where we recall that $W^{(1)}=W+\rho_W^{(1)} \star |\cdot|^{-1}$ and $W^{(k)}=\rho_W^{(k)} \star |\cdot|^{-1}$, we obtain
$$
\gamma_{\beta W} = \gamma_0 + \sum_{k=1}^{+\infty} Q^{(k)}\left(\sum_{j=1}^{+\infty} \beta^j W^{(j)} , \cdots , \sum_{j=1}^{+\infty} \beta^j W^{(j)}  \right),
$$
from which we deduce (\ref{eq:gammaWk}). Taking the densities of both sides of (\ref{eq:gammaWk}), we get
$$
\rho_W^{(k)} = - \cL (\rho_W^{(k)}) + \widetilde \rho_W^{(k)}.
$$
This proves (\ref{eq:CPrHF_rho}).

\subsection{Proof of Lemma~\ref{lem:RS_MO} and of~(\ref{eq:orthogonality})}
\label{sec:orthogonality}

The proof of Lemma~\ref{lem:RS_MO} is similar to the proof of Lemma~1 in~\cite{CanLeB98}. We only sketch it here for brevity. We denote by ${\cal V}:=(H^1(\R^3))^N$ , by $\Phi^0=(\phi_1^0,\cdots,\phi_N^0)^T\in {\cal V}$ and by $\mathscr H$ the bounded linear operator from ${\cal V}$ to ${\cal V}'\equiv(H^{-1}(\R^3))^N$ defined by
$$
\forall \Psi \in {\cal V}, \quad \left( {\mathscr H}\Psi\right)_i = (H_0-\epsilon_i)\psi_i + \sum_{j=1}^N K^0_{ij} \psi_j.
$$
We then decompose $\cal V$ as
$$
{\cal V}= {\mathbb S}\Phi^0+{\mathbb A}\Phi^0+\Phi^0_\perp = {\mathbb D}\Phi^0+{\mathbb S}^0\Phi^0+{\mathbb A}\Phi^0+\Phi^0_\perp,
$$
where ${\mathbb D}$, ${\mathbb A}$, ${\mathbb S}$, and ${\mathbb S}^0$ denote the vector spaces of $N \times N$ real-valued matrices which are respectively diagonal, antisymmetric, symmetric, and symmetric with zero entries on the diagonal, and where
$$
\Phi^0_\perp=\left\{\Phi = (\phi_i)_{1 \le i \le N} \in {\cal V} \; | \; \forall 1 \le i,j \le N, \; (\phi_i,\phi_j^0)_{L^2} =0\right\}.
$$
Likewise, it holds
$$
{\cal V}'= {\mathbb S}\Phi^0+{\mathbb A}\Phi^0+\Phi^0_{\perp\!\!\!\perp} \quad \mbox{with} \quad 
\Phi^0_{\perp\!\!\!\perp} =\left\{g = (g_i)_{1 \le i \le N} \in {\cal V}' \; | \; \forall 1 \le i,j \le N, \; \langle g_i,\phi_j^0\rangle =0\right\}
$$
and it is easily checked that
\begin{equation}\label{eq:ortagch}
\left\{g \in {\cal V}' \; | \; \forall \chi \in  \Phi^0_\perp, \; \langle g,\chi \rangle = 0 \right\} = {\mathbb S}\Phi^0+{\mathbb A}\Phi^0.
\end{equation}
Denoting by $F=(f_1,\cdots, f_N)^T \in {\cal V}'$ and by $\alpha \in {\mathbb D}$ the $N \times N$  diagonal matrix with entries $\alpha_1,\cdots, \alpha_N$, we have to show that there exists a unique pair $(\Psi,\eta) \in {\cal V}\times {\mathbb D}$ such that
\begin{equation}\label{eq:RS_MO_LS_2}
\left\{\begin{array}{l}
{\mathscr H}\Psi = F +\eta \Phi^0, \\
\Psi - \alpha \Phi^0 \in {\mathbb S}^0\Phi^0+{\mathbb A}\Phi^0+\Phi^0_\perp .
\end{array} \right.
\end{equation} 
For this purpose, we first introduce the matrix $S \in {\mathbb S}$ defined by
$$
\forall 1\le i \le N, \; S_{ii}=\alpha_{i}  \quad \mbox{and} \quad  \forall 1\le i \neq j \le N, \; S_{ij}=\frac{\langle f_j,\phi_i^0\rangle - \langle f_i,\phi_j^0\rangle}{\epsilon_j-\epsilon_i},
$$
and observe that $\widetilde F:=F-{\mathscr H}(S\Phi^0)   \in {\mathbb S}\Phi^0+\Phi^0_{\perp\!\!\!\perp}$.
Next, using  the fact that $\epsilon_1 < \cdots < \epsilon_N < \epsilon_{\rm F}^0$ and the positivity of the operator $K^0$, namely
$$
\forall \Psi=(\psi_i)_{1 \le i \le N} \in  {\cal V}, \quad \sum_{i,j=1}^N \langle K^0_{ij}\psi_j, \psi_i \rangle = 2\, D\left( \sum_{i=1}^N \phi_i^0\psi_i,\sum_{i=1}^N \phi_i^0\psi_i\right) \ge 0,
$$
we can see that the operator $\mathscr H$ is coercive on $\Phi^0_\perp$. Therefore, by Lax-Milgram lemma and (\ref{eq:ortagch}), there exists a unique $\widetilde \Psi \in  \Phi^0_\perp$ such that ${\mathscr H}\widetilde\Psi - \widetilde F\in {\mathbb S}\Phi^0+{\mathbb A}\Phi^0$. As $\widetilde F  \in {\mathbb S}\Phi^0+\Phi^0_{\perp\!\!\!\perp}$ and 
$$
\forall 1 \le i,k \le N, \quad \forall \Psi =(\psi_j)_{1 \le j \le N} \in {\cal V}, \quad \sum_{j=1}^N \langle K^0_{ij} \psi_j, \phi^0_k\rangle  = \sum_{j=1}^N \langle K^0_{kj} \psi_j, \phi^0_i\rangle, 
$$
we have in fact ${\mathscr H}\widetilde\Psi - \widetilde F\in {\mathbb S}\Phi^0$.
Setting $\Psi'=\widetilde\Psi+S\Phi^0$, we get ${\mathscr H}\Psi' - F \in {\mathbb S}\Phi^0$.  We now observe that 
${\mathscr H}$ is an isomorphism from ${\mathbb A}\Phi^0$ to ${\mathbb S}^0\Phi^0$. Decomposing ${\mathscr H}\Psi' - F$ as ${\mathscr H}\Psi' - F=-S'\Phi^0+ \eta \Phi^0$ with $S' \in  {\mathbb S}^0$ and $\eta \in {\mathbb D}$, and denoting by $A$ the unique element of ${\mathbb A}$ such that ${\mathscr H}(A\Phi^0)=S'\Phi^0$, and by $\Psi=\Psi'+A\Phi^0$, we finally obtain that the pair $(\Psi,\eta)$ is the unique solution to (\ref{eq:RS_MO_LS_2}) in ${\cal V}\times {\mathbb D}$.

The fact that $\Psi \in (H^2(\R^3))^N$ whenever $f \in (L^2(\R^3))^N$ follows from simple elliptic regularity arguments.

To prove~(\ref{eq:orthogonality}), we introduce, for $k \in \N^\ast$, 
$$
\chi_{i,k}(\beta) = \sum_{l=0}^k \beta^l \phi_{\beta W,i}^{(l)}, \quad 
\eta_{i,k}(\beta) = \sum_{l=0}^k \beta^l \epsilon_{\beta W,i}^{(l)}, 
$$
$$
H_k(\beta) = -\frac 12 \Delta + V + \left( \sum_{i=1}^N \chi_{i,k}(\beta)^2 \right) \star |\cdot|^{-1} + \beta W, \quad
f_{i,k}(\beta) = H_k(\beta) \chi_{i,k}(\beta) - \eta_{i,k}(\beta) \chi_{i,k}(\beta).
$$
By construction, $|\eta_{i,k}(\beta) - \epsilon_{\beta W,i} |  + \|\chi_{i,k}(\beta)-\phi_{\beta W,i}\|_{H^2}  + \|f_{i,k}(\beta)\|_{H^{-1}} \in {\cal O}(\beta^{k+1})$ when $\beta$ goes to zero, for all $1\le i \le N$. As the operator $H_k(\beta)$ is self-adjoint, it holds
$$
\langle  f_{i,k},\chi_{j,k}\rangle + \eta_{i,k} \langle  \chi_{i,k},\chi_{j,k}\rangle    = \langle H_k \chi_{i,k},\chi_{j,k}\rangle  = \langle H_k\chi_{j,k},\chi_{i,k}\rangle 
= \langle f_{j,k},\chi_{i,k}\rangle +  \eta_{j,k} \langle  \chi_{j,k},\chi_{i,k}\rangle
$$
(the variable $\beta$ has been omitted in the above equalities). As by assumption $\epsilon_1 < \epsilon_2 < \cdots < \epsilon_{N+1}$, we obtain
$$
\langle  \chi_{i,k}(\beta),\chi_{j,k}(\beta)\rangle = \frac{  \langle  f_{i,k}(\beta),\chi_{j,k}(\beta)\rangle  - \langle  f_{j,k}(\beta),\chi_{i,k}(\beta)\rangle   }{\eta_{j,k}(\beta)-\eta_{i,k}(\beta)}  \in {\cal O}(\beta^{k+1}),
$$
from which we deduce~(\ref{eq:orthogonality}).

\subsection{Proof of Lemma \ref{lem:projector_Pi}}

Let $T \in \Omega$ and $\gamma \in \cP_N$ such that $\|T-\gamma\|_{\gS_2} < 1/2$. As $\|T-\gamma\|\le \|T-\gamma\|_{\gS_2} < 1/2$, $\sigma(\gamma) = \left\{0,1\right\}$ and $\mbox{Rank}(\gamma)=N$, $\mbox{Rank}(\Pi(T))=\mbox{Rank}(\1_{[1/2,+\infty)}(T))=N$. Therefore $\Pi(T) \in \cP_N$. If, in addition, $T \in \gS_2$, then
\begin{eqnarray*}
\|T-\Pi(T)\|_{\gS_2}^2 &=& \|T-\gamma+\gamma-\Pi(T)\|_{\gS_2}^2 \\
&=& \|T-\gamma\|_{\gS_2}^2+\|\gamma-\Pi(T)\|_{\gS_2}^2+ 2 \tr\left((T-\gamma)(\gamma-\Pi(T))\right) \\
&=& \|T-\gamma\|_{\gS_2}^2+\|\gamma-\Pi(T)\|_{\gS_2}^2+ 2 \tr\left(T(\gamma-\Pi(T))\right) - \left( 2N-2\tr(\gamma\Pi(T)) \right) \\
&=& \|T-\gamma\|_{\gS_2}^2+ 2 \, \tr\left(T(\gamma-\Pi(T))\right) \\
&=& \|T-\gamma\|_{\gS_2}^2+ 2 \, \tr\left((T-1/2)(\gamma-\Pi(T))\right),
\end{eqnarray*} 
where we have used that both $\gamma$ and $\Pi(T)$ are in ${\cal P}_N$ and that for all $P \in \cP_N$, $\|P\|_{\gS_2}^2=\tr(P^2)=\tr(P)=N$. Let $A=T-1/2$ and $Q=\gamma-\Pi(T)$. The self-adjoint operator $A$ has exactly $N$ positive eigenvalues (counting multiplicities), and all its other eigenvalues are negative. Remarking that $\Pi(T) = \1_{[0,+\infty)}(A)$, and denoting $A^{++} = \Pi(T)A\Pi(T)$, $A^{--}=(1-\Pi(T))A(1-\Pi(T))$, $Q^{--}=\Pi(T)(\gamma-\Pi(T))\Pi(T)$, $Q^{++}=(1-\Pi(T))(\gamma-\Pi(T))(1-\Pi(T))$, and $g:= \mbox{\rm dist}(0,\sigma(A))$, we obtain, using the fact that $A^{++}\ge g$, $A^{--} \le -g$, $Q^{++}\ge 0$, $Q^{--}\le 0$ and $Q^2=Q^{++}-Q^{--}$,  
\begin{eqnarray*}
\tr\left((T-1/2)(\gamma-\Pi(T))\right)&=&\tr(A^{++}Q^{--}+A^{--}Q^{++}) \\
&\le& -g \tr(Q^{++}-Q^{--}) = -g \tr(Q^2) = -g \|\gamma-\Pi(T)\|_{\gS_2}^2.
\end{eqnarray*}
Hence, $\Pi(T)$ is the unique minimizer of (\ref{eq:projector_Pi}).

\subsection{Proof of Theorem \ref{th:Wigner_non_degenerate}}

Throughout the proof, $W$ is a fixed potential of  ${\cal C}'$, chosen once and for all, and $C$ denotes a constant depending on $W$ but not on $\beta$, which may vary from one line to another. For all $\beta \in \R$, we denote by $Q^{(n)}_W(\beta):=\widetilde \gamma^{(n)}_W(\beta)-\gamma_{\beta W}$.
When $|\beta|$ is small enough, $\widetilde \gamma^{(n)}_W(\beta) \in \cP_N$, so that we have
\begin{eqnarray*}
E^{\rm rHF}(\widetilde \gamma^{(n)}_W(\beta),\beta W) & \ge & \cE^{\rm rHF}(\beta W) \\ 
&=& E^{\rm rHF}(\gamma_{\beta W},\beta W) \\
&=& E^{\rm rHF}(\widetilde \gamma^{(n)}_W(\beta)- Q^{(n)}_W(\beta),\beta W) \\
&=& E^{\rm rHF}(\widetilde \gamma^{(n)}_W(\beta),\beta W)-\tr\left( H_{\beta W} Q^{(n)}_W(\beta)\right)-\frac 12 D\left(\rho_{Q^{(n)}_W(\beta)},\rho_{Q^{(n)}_W(\beta)}\right) \\
&=& E^{\rm rHF}(\widetilde \gamma^{(n)}_W(\beta),\beta W)-\tr\left( |H_{\beta W}-\epsilon_{\rm F}^0| (Q^{(n)}_W(\beta))^2\right) - \frac 12 \|\rho_{Q^{(n)}_W(\beta)}\|_\cC^2,
\end{eqnarray*}
where we have used Lemma~\ref{lem:upper_bound} below. We thus obtain that for $|\beta|$ small enough,
$$
0 \le E^{\rm rHF}(\widetilde \gamma^{(n)}_W(\beta),\beta W) -\cE^{\rm rHF}(\beta W) = \tr\left( |H_{\beta W}-\epsilon_{\rm F}^0| (Q^{(n)}_W(\beta))^2\right) + \frac 12 \|\rho_{Q^{(n)}_W(\beta)}\|_\cC^2.
$$
Using (\ref{eq:HW2}), (\ref{eq:norm_vW}) and the bound $\|v(1-\Delta)^{-1}\| \le C \|v\|_{\cC'}$ for all $v \in \cC'$, we obtain that for all $|\beta|$ small enough,
$$
|H_{\beta W}-\epsilon_{\rm F}^0| \le C(1-\Delta).
$$
Hence, for $|\beta|$ small enough,
\begin{eqnarray*}
0 \le E^{\rm rHF}(\widetilde \gamma^{(n)}_W(\beta),\beta W) -\cE^{\rm rHF}(\beta W) &\le &C \tr\left(  (1-\Delta) (Q^{(n)}_W(\beta))^2\right) + \frac 12 \|\rho_{Q^{(n)}_W(\beta)}\|_\cC^2 \\ &  \le &   C \|Q^{(n)}_W(\beta)\|_{\gS_{1,1}}^2,
\end{eqnarray*}
where we have used the continuity of the linear mapping $\gS_{1,1} \ni \gamma \mapsto \rho_\gamma \in {\cal C}$. The latter property is proved as followed: we infer from the Kato-Seiler-Simon inequality and the Sobolev inequality $\|V\|_{L^6(\R^3)} \le C_6 \|\nabla V\|_{L^2(\R^3)}= C_6\|V\|_{\cC'}$ that there exists a constant $C \in \R_+$ such that for all $\gamma \in \gS_{1,1}\cap {\cal S}(L^2(\R^3))$,
\begin{eqnarray}
\|\rho_\gamma\|_{\cC} &=& \sup_{V \in \cC'\setminus\left\{0\right\}} \frac{\tr(\gamma V)}{\|V\|_{\cC'}}
= \sup_{V \in \cC'\setminus\left\{0\right\}} \frac{\tr((1-\Delta)^{1/2}\gamma(1-\Delta)^{1/2} (1-\Delta)^{-1/2} V(1-\Delta)^{-1/2})}{\|V\|_{\cC'}} \nonumber \\ &\le& C \|\gamma\|_{\gS_{1,1}}. \label{eq:boundRho}
\end{eqnarray}
Denoting by 
$$
\gamma_{W,n}(\beta):= \gamma_0+\sum_{k=1}^n \beta^k \gamma_W^{(k)},
$$
we get 
$$
0 \le E^{\rm rHF}(\widetilde \gamma^{(n)}_W(\beta),\beta W) -\cE^{\rm rHF}(\beta W) \le C \left(  \| \widetilde\gamma_W^{(n)}(\beta) - \gamma_{W,n}(\beta)\|_{\gS_{1,1}}^2 + \| \gamma_{W,n}(\beta) - \gamma_{\beta W}\|_{\gS_{1,1}}^2 \right).
$$
We infer from the third statement of Theorem~\ref{th:CPrHF} that 
$$
\| \gamma_{W,n}(\beta) - \gamma_{\beta W}\|_{\gS_{1,1}} \le C \beta^{n+1}.
$$
We now observe that as $W$ is fixed, all the functions $\widetilde \phi_{W,i}(\beta)$ in (\ref{eq:gammaWn})-(\ref{eq:tgammaWn}) lay in a finite dimensional subspace of $H^1(\R^3)$ independent of $\beta$. Using the equivalence of norms in finite dimension, the fact that $\widetilde\gamma_W^{(n)}(\beta)=\Pi\left( \gamma_{W,n}(\beta)\right)$ and Lemma~\ref{lem:projector_Pi}, we obtain that 
$$
\| \widetilde\gamma_W^{(n)}(\beta) - \gamma_{W,n}(\beta)\|_{\gS_{1,1}} \le  C \| \widetilde\gamma_W^{(n)}(\beta) - \gamma_{W,n}(\beta)\|_{\gS_{2}} \le C \|   \gamma_{\beta W} - \gamma_{W,n}(\beta)\|_{\gS_{2}} \le C \beta^{n+1},
$$
which completes the proof of (\ref{eq:Wigner}).

\medskip

\begin{lemma} \label{lem:upper_bound} Let $H$ be a bounded below self-adjoint operator on a Hilbert space ${\cal H}$, $\epsilon_{\rm F} \in \R$, and $\gamma:=\1_{(-\infty,\epsilon_{\rm F}]}(H)$. Assume that $\tr(\gamma) < \infty$. Then, for all orthogonal projector $\gamma' \in {\cal S}({\cal H})$ such that $\tr(\gamma')=\tr(\gamma)$, it holds
$$
0 \le \tr(HQ) = \tr(|H-\epsilon_{\rm F}|Q^2),
$$
where $Q=\gamma'-\gamma$.
\end{lemma}

\medskip

\begin{proof} We first observe that
$$
Q=\gamma'-\gamma = (\gamma')^2-\gamma^2=Q^2 + \gamma\gamma'+\gamma'\gamma-2\gamma,
$$
$$
H-\epsilon_{\rm F} =  (1-\gamma) (H-\epsilon_{\rm F}) (1-\gamma) + \gamma (H-\epsilon_{\rm F}) \gamma,
$$$$
|H-\epsilon_{\rm F}| =  (1-\gamma) (H-\epsilon_{\rm F}) (1-\gamma) - \gamma (H-\epsilon_{\rm F}) \gamma,
$$
$$
Q^2=(1-\gamma) Q (1-\gamma) - \gamma Q \gamma.
$$
As $\tr(Q)=0$, it follows that
\begin{eqnarray*}
\tr\left(HQ\right)&=& \tr((H-\epsilon_{\rm F})Q)=\tr\left((H-\epsilon_{\rm F})Q^2\right)+ 
\tr\left((H-\epsilon_{\rm F})(\gamma\gamma'+\gamma'\gamma-2\gamma)\right) \\
&=&\tr\left((H-\epsilon_{\rm F})Q^2\right)+ 
2\, \tr\left(\gamma (H-\epsilon_{\rm F}) \gamma Q\right) \\
&=& \tr\left((H-\epsilon_{\rm F})Q^2\right)+2\, \tr\left( \gamma (H-\epsilon_{\rm F}) \gamma Q \gamma \right) \\
&=& \tr\left((H-\epsilon_{\rm F})Q^2\right)- 2\, \tr\left(\gamma(H-\epsilon_{\rm F})\gamma Q^2\right) \\
&=& \tr\left(|H-\epsilon_{\rm F}|Q^2\right).
\end{eqnarray*}
Note that all the terms in the above series of equalities containing $\gamma$ are finite, since $\tr(\gamma)<\infty$ and $H$ is bounded below, while the other terms may be equal to $+\infty$.
\end{proof}

\subsection{Proof of Lemma~\ref{lem:localmap}}

\noindent
Using the fact that $L^2(\R^3)=\cH_{\rm o}\oplus\cH_{\rm u}$, any linear operator $T$ on $L^2(\R^3)$ can be represented by a $2\times 2$ block operator
$$
T = \left( \begin{array}{cc} T_{\rm oo} &  T_{\rm ou} \\
T_{\rm uo} &  T_{\rm uu} \end{array}\right),
$$ 
where $T_{\rm xy}$ is a linear operator from $\cH_{\rm y}$ to $\cH_{\rm x}$ (with $x,y \in \left\{{\rm o},{\rm u}\right\}$). In particular, the operators $P_0:=\1_{(-\infty,\epsilon_{\rm F}^0]}(H_0)$ (the orthogonal projector on $\cH_{\rm o}$), $P_0^\perp:=\1_{(\epsilon_{\rm F}^0,+\infty)}(H_0)$ and $H_0$ are block diagonal in this representation, and we have
$$
P_0 = \left( \begin{array}{cc} 1 &  0 \\ 0 &  0 \end{array}\right), \qquad 
P_0^\perp = \left( \begin{array}{cc} 0 &  0 \\ 0 &  1 \end{array}\right), \qquad
H_0 = \left( \begin{array}{cc} H_{\rm oo} &  0 \\ 0 &   H_{\rm uu} \end{array}\right),
$$
with $H_{\rm oo}-\epsilon_{\rm F}^0 \le 0$ and $H_{\rm uu}-\epsilon_{\rm F}^0 = H_0^{++}-\epsilon_{\rm F}^0  \ge g_+ > 0$.

\medskip

We consider the submanifold
$$
\cP_{N_{\rm o}}:= \left\{ P \in \cS(L^2(\R^3)) \; | \; P^2=P, \; \tr(P)=N_{\rm o}, \; \tr(-\Delta P) < \infty \right\} 
$$
of $\cS(L^2(\R^3))$ consisting of the rank-$N_{\rm o}$ orthogonal projectors on $L^2(\R^3)$ with range in $H^1(\R^3)$, and the Hilbert space
$$
\cZ=\left\{Z = \left( \begin{array}{cc} 0 & -  Z_{\rm uo}^\ast \\   Z_{\rm uo} & 0 \end{array}\right) \; | \;   (H_{\rm uu}-\epsilon_{\rm F}^0)^{1/2}Z_{\rm uo} \in {\cal B}(\cH_{\rm o},\cH_{\rm u})\right\},
$$
endowed with the inner product
$$
(Z,Z')_\cZ=\tr(Z_{\rm uo}^\ast(H_{\rm uu}-\epsilon_{\rm F}^0)Z'_{\rm uo}).
$$
We are going to use the following lemma, the proof of which is postponed until the end of the section.

\medskip

\begin{lemma}\label{lem:loc_map} There exists an open connected neighborhood $\widetilde\cO$ of $P_0$ in $\cP_{N_{\rm o}}$, and $\eta > 0$ such that the real analytic mapping 
\begin{eqnarray*}
B_\eta(\cZ) & \rightarrow & \widetilde\cO \\
Z & \mapsto & e^{Z} P_0 e^{-Z}
\end{eqnarray*}
is bijective.
\end{lemma}

\medskip

By continuity, there exists a neighborhood $\cO$ of $0$ in $\cA$ such that 
$$
\forall A \in \cO, \quad \1_{(0,1]}\left(\Gamma(A)\right) \subset  \widetilde\cO.
$$
Let $A$ and $A'$ in $\cO$ be such that $\Gamma(A)=\Gamma(A')$. Then
$$
e^{L_{\rm uo}(A')} P_0 e^{-L_{\rm uo}(A')}=\1_{(0,1]}\left(\Gamma(A')\right) =\1_{(0,1]}\left(\Gamma(A)\right) =e^{L_{\rm uo}(A)} P_0 e^{-L_{\rm uo}(A)},
$$
and we infer from Lemma~\ref{lem:loc_map} that $L_{\rm uo}(A')=L_{\rm uo}(A)$. Therefore,
\begin{equation}\label{eq:LpfLpp}
e^{L_{\rm pf}(A')} (\gamma_0+L_{\rm pp}(A'))  e^{-L_{\rm pf}(A')} = e^{L_{\rm pf}(A)} (\gamma_0+L_{\rm pp}(A))  e^{-L_{\rm pf}(A)}.
\end{equation}
In particular (using again functional calculus),
$$
e^{L_{\rm pf}(A')} \gamma_0 e^{-L_{\rm pf}(A')} = e^{L_{\rm pf}(A)} \gamma_0 e^{-L_{\rm pf}(A)}.
$$
Using the finite dimensional analogue of Lemma~\ref{lem:loc_map} (a standard result on finite dimensional Grassmann manifolds), we obtain that, up to reducing the size of the neighborhood $\cO$ if necessary, $L_{\rm pf}(A')=L_{\rm pf}(A)$. Getting back to (\ref{eq:LpfLpp}), we see that $L_{\rm pp}(A')=L_{\rm pp}(A)$. Therefore, $A=A'$, which proves the injectivity of the mapping (\ref{eq:localchart}).

\medskip

We now consider a neighborhood $\cO'$ of $\gamma_0$ in $\gS_{1,1}$ in such that $\Gamma(\cO)\subset \cO'$ and $\1_{(0,1]}\left(\cK_{N_{\rm f},N_{\rm p}} \cap \cO' \right) \subset  \widetilde\cO$. Let $\gamma \in \cK_{N_{\rm f},N_{\rm p}} \cap \cO'$. By Lemma~\ref{lem:loc_map}, there exists a unique $Z \in B_\eta(\cZ)$ such that $\1_{(0,1]}(\gamma) = e^{Z} P_0 e^{-Z}$, and by the classical finite-dimensional version of the latter lemma, there exists a unique $A_{\rm pf} \in \cA_{\rm pf}$ in the vicinity of $0$ such that   $\1_{\left\{1\right\}}(\gamma) = e^{Z} e^{L_{\rm pf}(0,0,A_{\rm pf},0)} \1_{\left\{1\right\}}(\gamma_0) e^{-L_{\rm pf}(0,0,A_{\rm pf},0)} e^{-Z}$. It is then easily seen that the operator 
$$
e^{-Z}e^{-L_{\rm pf}(0,0,A_{\rm pf},0)} \gamma e^{L_{\rm pf}(0,0,A_{\rm pf},0)} e^{Z}
$$
is of the form $\gamma_0+L_{\rm pp}(0,0,0,A_{\rm pp})$ for some $A_{\rm pp} \in \cA_{\rm pp}$, which is close to $0$ if $\cO'$ is small enough. Decomposing $Z_{\rm uo}$ as $(A_{\rm uf},A_{\rm up})$ and setting $A=(A_{\rm uf},A_{\rm up},A_{\rm pf},A_{\rm pp})$, we obtain that $A$ is the unique element of $\cA$ in the vicinity of $0$ such that $\gamma = \Gamma(A)$.

\medskip

\begin{proof}[Proof of Lemma~\ref{lem:loc_map}] Let
$$
{\cal U} := \left\{ U \in \mbox{GL}(H^1(\R^3)) \; | \; \|U\phi\|_{L^2}=\|\phi\|_{L^2}, \; \forall \phi \in H^1(\R^3)\right\}
$$
where $\mbox{GL}(H^1(\R^3))$ is the group of the inversible bounded operators on $H^1(\R^3)$.
In view of~\cite[Theorem 4.8]{ChiMel12}, the mapping 
\begin{eqnarray*}
{\cal U} & \rightarrow & \cP_{N_{\rm o}}  \\
U & \mapsto & U P_0 U^{-1}
\end{eqnarray*}
is a real analytic submersion. Besides \cite[Lemma 2.5]{ChiMel12}, $\cal U$ is a Banach-Lie group with Lie algebra
$$
{\mathscr U}=\left\{ Z \in \cB(L^2(\R^3)) \; | \; Z^\ast=-Z, \; Z(H^1(\R^3)) \subset H^1(\R^3)  \right\}
$$
(with the slight abuse of notation consisting of denoting by $Z$ the restriction to $H^1(\R^3)$ of an operator $Z \in \cB(L^2(\R^3))$ such that $Z(H^1(\R^3)) \subset H^1(\R^3)$), and  \cite[Remark 4.7]{ChiMel12}, the isotropy group of the action of $\cal U$ on $\cP_{N_{\rm o}} $ is the Banach-Lie group with Lie algebra 
$$
{\mathscr U}_0=\left\{ Z \in \cB(L^2(\R^3)) \; | \; Z^\ast=-Z, \; Z(H^1(\R^3)) \subset H^1(\R^3), \; Z_{\rm uo}=0  \right\}.
$$
Hence, denoting by
$$
\widetilde \cZ= \left\{Z = \left( \begin{array}{cc} 0 & -  Z_{\rm uo}^\ast \\   Z_{\rm uo} & 0 \end{array}\right) \; | \;    (1-\Delta)^{1/2} Z_{\rm uo} \in {\cal B}(\cH_{\rm o},\cH_{\rm u})\right\},
$$
there exists an open connected neighborhood $\widetilde\cO$ of $P_0$ in $\cP_{N_{\rm o}}$, and $\widetilde \eta > 0$ such that the real analytic mapping 
\begin{eqnarray*}
B_{\widetilde\eta}(\widetilde \cZ) & \rightarrow & \widetilde\cO \\
Z & \mapsto & e^{Z} P_0 e^{-Z}
\end{eqnarray*}
is bijective. As there exists $0 < c < C < \infty$ such that $c (1-\Delta) \le (H_{\rm uu}-\epsilon_{\rm F}^0) \le C (1-\Delta)$ on ${\cal H}_{\rm u}$, we have $\widetilde \cZ=\cZ$, which concludes the proof of the lemma.
\end{proof}

\subsection{Proof of Lemma~\ref{lem:preIFT}}

\noindent
In view of (\ref{eq:expGA}), the density matrix $\Gamma(A)$ can be expanded as
\begin{equation} \label{eq:expandGamma}
 \Gamma(A)=\gamma_0+\gamma_1(A)+\gamma_2(A,A)+ O(\|A\|_{\mathcal V}^3),  
\end{equation}
with
\begin{eqnarray*}
\gamma_1(A)&=& \langle \Gamma'(0),A \rangle = [L_{\rm uo}(A)+L_{\rm pf}(A),\gamma_0]+L_{\rm pp}(A)  \\
\gamma_2(A,A)&=&  \frac 12 [\Gamma''(0)](A,A) \\
&=&  \frac 12 \left[ L_{\rm uo}(A),\left[ L_{\rm uo}(A),\gamma_0\right]\right] +   \left[ L_{\rm uo}(A),\left[ L_{\rm pf}(A),\gamma_0\right]\right] +    \frac 12 \left[ L_{\rm pf}(A),\left[ L_{\rm pf}(A),\gamma_0\right]\right] \\ && + \left[ L_{\rm uo}(A),L_{\rm pp}(A)\right]  + \left[ L_{\rm pf}(A),L_{\rm pp}(A)\right]   \\
&=&\frac 12 \left\{ L_{\rm uo}(A)^2+L_{\rm pf}(A)^2,\gamma_0 \right\} + [L_{\rm uo}(A)+L_{\rm pf}(A),L_{\rm pp}(A)] \\
&& + L_{\rm uo}(A)L_{\rm pf}(A)\gamma_0+\gamma_0L_{\rm pf}(A)L_{\rm uo}(A)- (L_{\rm uo}(A)+L_{\rm pf}(A))\gamma_0(L_{\rm uo}(A)+L_{\rm pf}(A)),
\end{eqnarray*}
where $\left\{ X,Y \right\}=XY+YX$ denotes the anticommutator of $X$ and $Y$.
As in Section~\ref{sec:bounded}, we denote by  $F(A,0)=\nabla_A E(A,0)$ and $\Theta=\frac 12 F'_A(0,0)|_{{\cal A} \times \left\{0\right\}}$. It follows from (\ref{eq:expandGamma}) and the analyticity properties of the mapping $A \mapsto E(A,0)$ that for all $(A,A') \in {\cal A} \times {\cal A}$,
$$
E(A,0)= E_0+\tr(H_0\gamma_1(A))+\tr(H_0\gamma_2(A,A))+\frac{1}{2}D(\rho_{\gamma_1(A)},\rho_{\gamma_1(A)})+ O(\|A\|_{\mathcal A}^3),
$$
and
$$
\langle\Theta(A),A\rangle=\tr(H_0\gamma_2(A,A))+\frac 12 D(\rho_{\gamma_1(A)},\rho_{\gamma_1(A)}).
$$
Besides, a simple calculation leads to 
\begin{eqnarray*}
\tr(H_0\gamma_2(A,A))&=&\tr\left(A_{\rm uf}^\ast\left(H_0^{++}-\epsilon_{\rm F}^0\right) A_{\rm uf}\right)-\tr\left(A_{\rm uf}\left(H_0^{--}-\epsilon_{\rm F}^0\right) A_{\rm uf}^\ast\right) \\
&& + \tr\left(\left(H_0^{++}-\epsilon_{\rm F}^0\right)A_{\rm up}\Lambda A_{\rm up}^\ast\right) - \tr\left(\left(H_0^{--}-\epsilon_{\rm F}^0\right)A_{\rm pf}^\ast(1-\Lambda) A_{\rm pf}\right). 
\end{eqnarray*}
Hence,
\begin{equation} \label{eq:ThetaAA'}
\langle\Theta(A),A' \rangle=a(A,A')+\frac 12 D(\rho_{\gamma_1(A)},\rho_{\gamma_1(A')}),
\end{equation}
where
\begin{eqnarray*}
a(A,A')&=&\tr\left(A_{\rm uf}^\ast\left(H_0^{++}-\epsilon_{\rm F}^0\right) A_{\rm uf}'\right)-\tr\left(A_{\rm uf}'\left(H_0^{--}-\epsilon_{\rm F}^0\right) A_{\rm uf}^\ast\right) \\
&& + \tr\left(\left(H_0^{++}-\epsilon_{\rm F}^0\right)A_{\rm up}'\Lambda A_{\rm up}^\ast\right) - \tr\left(\left(H_0^{--}-\epsilon_{\rm F}^0\right)A_{\rm pf}^\ast(1-\Lambda) A_{\rm pf}'\right).
\end{eqnarray*}
For all $A$ and $A'$ in ${\cal A}$, we have
$$
|a(A,A')|\le \left(1+\frac{\epsilon_{\rm F}^0-\epsilon_1}{g_+}\right)\|A_{\rm uf}\|_{{\cal A}_{\rm uf}}\|A_{\rm uf}'\|_{{\cal A}_{\rm uf}}
+ \|A_{\rm up}\|_{{\cal A}_{up}}\|A_{\rm up}'\|_{{\cal A}_{\rm up}} + (\epsilon_{\rm F}^0-\epsilon_1) \|A_{\rm pf}\|_{{\cal A}_{\rm pf}}\|A_{\rm pf}'\|_{{\cal A}_{\rm pf}}.
$$
We thus deduce from (\ref{eq:boundRho}) that there exists a constant $C' \in \R_+$ such that for all $A \in {\cal A}$,
$$
\|\rho_{\gamma_1(A)}\|_{\cC} \le C \|\gamma_1(A)\|_{\gS_{1,1}} \le C' \|A\|_{\cal A}.
$$
The bilinear form in (\ref{eq:ThetaAA'}) is therefore continuous on the Hilbert space ${\cal A}$. It is also positive since for all $A \in {\cal A}$,
\begin{equation}\label{eq:thetaAA}
\langle\Theta(A),A\rangle \ge  \|A_{\rm uf}\|_{{\cal A}_{\rm uf}}^2+\lambda_- \|A_{\rm up}\|_{{\cal A}_{\rm up}}^2+(1-\lambda_+) g_-\|A_{\rm pf}\|_{{\cal A}_{\rm pf}}^2 + \frac 12 \|\rho_{\gamma_1(A)}\|_{\cC}^2,
\end{equation}
where $0 < \lambda_- \le \lambda_+ < 1$ are the lowest and highest eigenvalues of $\Lambda$. To prove that it is in fact coercive, we proceed by contradiction and assume that there exists a normalized sequence $(A_k)_{k \in \N}$ in ${\cal A}$ such that $\lim_{k \to \infty}\langle\Theta(A_k),A_k\rangle=0$. We infer from (\ref{eq:thetaAA}) that $\|(A_k)_{\rm uf}\|_{{\cal A}_{\rm uf}}$, $\|(A_k)_{\rm up}\|_{{\cal A}_{\rm up}}$, $\|(A_k)_{\rm pf}\|_{{\cal A}_{\rm pf}}$ and $\|\rho_{\gamma_1(A_k)}\|_{\cC}$ converge to zero when $k$ goes to infinity. Denoting by $(M_k)_{ij} := (\phi_{N_{\rm f}+i}^0, (A_k)_{\rm pp}\phi_{N_{\rm f}+j}^0)_{L^2}$, this implies that $\|M_k\|_2= \|(A_k)_{\rm pp}\|_{{\cal A}_{\rm pp}}\rightarrow 1$ and
$$
\left\| \sum_{i,j=1}^{N_{\rm p}} (M_k)_{ij} \phi_{N_{\rm f}+i}^0\phi_{N_{\rm f}+j}^0 \right\|_\cC \rightarrow 0.
$$
Extracting from $(M_k)_{k \in \N}$ a subsequence $(M_{k_n})_{n \in \N}$ converging to some $M \in \R_{\rm S}^{N_{\rm p}\times N_{\rm p}}$, and letting $n$ go to infinity, we obtain 
$$
\|M\|_2=1 \qquad \mbox{and} \qquad \sum_{i,j=1}^{N_{\rm p}} M_{ij} \phi_{N_{\rm f}+i}^0\phi_{N_{\rm f}+j}^0 =0.
$$
This contradicts (\ref{eq:hyp_ND}). The bilinear form (\ref{eq:ThetaAA'}) is therefore coercive on ${\cal A}$. As it is also continuous, we obtain that the linear map $\Theta$ is a bicontinuous coercive isomorphism from ${\cal A}$ to ${\cal A}'$.

\subsection{Proof of Lemma~\ref{lem:min_KNN=min_K}}

We can prove the existence of a minimizer $\widetilde\gamma_W$ to (\ref{eq:min_rHF}) reasoning as in the proof of the first statement of Theorem~\ref{th:CPrHF} (non-degenerate case) up to (\ref{eq:norm_vW}). Only the final argument is slightly different. In the degenerate case, we deduce that $H_W$ has at least $N$ negative eigenvalues from the fact that $\mbox{Rank}(\1_{(-\infty,\alpha_5]}(H_W))=N_{\rm o} \ge N$. 

\medskip

We now have to prove that $\widetilde \gamma_W=\gamma_W$, where $\gamma_W$ is defined by (\ref{eq:def_gamma_W}). We know that $\gamma_W$ is the unique local minimizer of (\ref{eq:min_rHF_KNN}) in the neighborhood of $\gamma_0$. Decomposing the space $L^2(\R^3)$ as 
\begin{equation}\label{eq:decompL2W}
L^2(\R^3) = \cH_{\rm f}^W \oplus \cH_{\rm p}^W \oplus \cH_{\rm u}^W, 
\end{equation}
where $\cH_{\rm f}^W=\mbox{Ran}(\1_{\left\{1\right\}}(\gamma_W))$, $\cH_{\rm p}^W=\mbox{Ran}(\1_{(0,1)}(\gamma_W))$, and $\cH_{\rm u}^W=\mbox{Ran}(\1_{\left\{0\right\}}(\gamma_W))$, we can parametrize $\cK_{N_{\rm f},N_{\rm p}}$ in the neighborhood of  $\gamma_W$ using the local map
$$
\Gamma^W(A) := \exp\left(L_{\rm uo}^W(A)\right) \; \exp\left(L_{\rm pf}^W(A)\right) \; \left( \gamma_W+L_{\rm pp}^W(A)\right) \; \exp\left(-L_{\rm pf}^W(A)\right) \; \exp\left(-L_{\rm uo}^W(A)\right),
$$
where
$$
L_{\rm uo}^W(A) := \left[\begin{matrix}
    0 & 0 & -A_{\rm uf}^\ast  \\ \\   
    0  &  0   &   -A_{\rm up}^\ast  \\ \\
    A_{\rm uf}   &  A_{\rm up}   &  0  
   \end{matrix}\right] , \quad L_{\rm pf}^W(A) := \left[\begin{matrix}
    0 & -A_{\rm pf}^\ast & 0  \\ \\   
    A_{\rm pf}  &  0   &  0 \\ \\
    0   &  0   &  0  
   \end{matrix}\right] , \quad L_{\rm pp}^W(A) := \left[\begin{matrix}
    0 & 0 & 0  \\ \\   
    0  &  A_{\rm pp}   &  0 \\ \\
    0   &  0   &  0  
   \end{matrix}\right],
$$
the block decomposition of the operators $L_{\rm xy}^W(A)$ being done with respect to the decomposition (\ref{eq:decompL2W}) of the space $L^2(\R^3)$. As $A=0$ is the unique minimizer of the functional $A \mapsto \cE^{\rm rHF}(\Gamma^W(A),W)$ in the neighborhood of $0$, we obtain that the block decomposition of the operator $\widetilde H= \dps -\frac 12 \Delta + V + \rho_{\gamma_W} \star |\cdot|^{-1}+ W$ reads
$$
\widetilde H  := \left[\begin{matrix}
  \widetilde H_{\rm ff}    & 0 &  0 \\ \\   
    0  &  \widetilde H_{\rm pp}     &   0 \\ \\
    0  &  0   &  \widetilde H_{\rm uu}    
   \end{matrix}\right] 
$$
(first-order optimality conditions), and that there exists $\epsilon \in \R$ such that
$$
\widetilde H_{\rm ff} -\epsilon \le 0, \quad \widetilde H_{\rm pp} -\epsilon = 0, \quad \widetilde H_{\rm uu} -\epsilon \ge 0 
$$
(second-order optimality conditions). These conditions also read
\begin{equation}\label{eq:EEdeg}
\gamma_W = \1_{(-\infty,\epsilon)}(\widetilde H)+ \delta_W,
\end{equation}
with $0 \le \delta_W \le 1$, $\mbox{Ran}(\delta_W) \subset \mbox{Ker}(\widetilde H-\epsilon)$, $\tr(\gamma_W)=N$, which are precisely  the Euler conditions for problem (\ref{eq:min_rHF}). Thus, $\gamma_W$ is a minimizer to (\ref{eq:min_rHF}).

\medskip

It follows that all the minimizers $\widetilde \gamma_W$ of (\ref{eq:min_rHF}) have density $\rho_W:=\rho_{\gamma_W}$ and are of the form 
$$
\widetilde \gamma_W = \1_{(-\infty,\epsilon)}(\widetilde H)+ \widetilde \delta_W,
$$
with $0 \le \widetilde \delta_W \le 1$, $\mbox{Ran}(\widetilde \delta_W) \subset \mbox{Ker}(\widetilde H-\epsilon)$, $\tr(\widetilde \gamma_W)=N$. As the optimization problem (\ref{eq:min_rHF}) is convex, the set of its minimizers is convex. Therefore, for any $t \in [0,1]$
$$
(1-t) \gamma_W + t  \widetilde \gamma_W = \1_{(-\infty,\epsilon)}(\widetilde H)+ (1-t) \delta_W+ t \widetilde \delta_W,
$$
is a global minimizer of (\ref{eq:min_rHF}), hence of  (\ref{eq:min_rHF_KNN}) for $t$ small enough. As we know that $\gamma_W$ is the unique minimizer to (\ref{eq:min_rHF_KNN})  in the vicinity of $\gamma_0$, we obtain that $\widetilde \delta_W=\delta_W$, which proves that $\gamma_W$ is the unique minimizer of  (\ref{eq:min_rHF}).

\subsection{Proof of~Theorem~\ref{th:rHF}}

The first statement of~Theorem~\ref{th:rHF} has been proved in the previous section. The second statement is a consequence of (\ref{eq:EEdeg}) and of the fact that $\gamma_W \in \cK_{N_{\rm f},N_{\rm p}}$. The third statement follows from the real analyticity of the mappings $B_\eta(\cC') \ni W \mapsto \widetilde A(W) \in {\cal A}$, ${\cal A} \ni A \mapsto \Gamma(A) \in \gS_{1,1}$, and $\gS_{1,1} \times \cC' \ni (\gamma,W) \mapsto E^{\rm rHF}(\gamma,W) \in \R$ and the chain rule.
 
\medskip

\noindent
It follows from (\ref{eq:expGA}) that for all $A\in\cal O$ and all $W\in \cal C'$,
\begin{eqnarray*}
 E(A,W)&=&E_0+\int_{\R^3}\rho_{\gamma_0}W+ \langle \Theta(A),A \rangle +\int_{\R^3}\rho_{\gamma_1(A)}W+\sum_{l\geq 3} \tr(H_0\gamma_l(A,\cdots,A)) \\
       &+& \frac 12 \sum_{\underset{l,l'\geq 1}{l+l'\geq 3}}D(\rho_{\gamma_{l}(A,\cdots,A)},\rho_{\gamma_{l'}(A,\cdots,A)})+\sum_{l\geq 2}\int_{\R^3}\rho_{\gamma_{l}(A,\cdots,A)}W.
\end{eqnarray*}
As a consequence, we obtain that that for any $A'\in\cal O$, 
\begin{eqnarray}\label{eq:expenergy}
(\nabla_A E(A,W),A' )_{\cal A} &=& 2\langle \Theta(A),A' \rangle +\int_{\R^3}\rho_{\gamma_1(A')}W+\sum_{l\geq 3} \tr(H_0\Gamma_l(A,A'))\nonumber \\ 
       && + \sum_{\underset{l\geq 1,\; l'\geq 1}{l+l'\geq 3}}D(\rho_{\gamma_{l}(A,\cdots, A)},\rho_{\Gamma_{l'}(A,A')})+\sum_{l\geq 2}\int_{\R^3}\rho_{\Gamma_{l}(A,A')}W,
\end{eqnarray}
with where $\Gamma_1(A,A')=\gamma_1(A')$ is in fact independent of $A$, and where for all $l \ge 2$, $\Gamma_{l}(A,A')=\sum_{i=1}^{l}\gamma_l(\tau_{(i,l)}(A,\cdots,A,A'))$ (recall that $\tau_{(i,l)}$ denotes the transposition swapping the $i^{\rm th}$ and $l^{\rm th}$ elements, and that, by convention $\tau_{l,l}$ is the identity).  By definition of $A_W(\beta)$, we have 
\begin{equation}\label{eq:zerograd}
\forall A' \in {\cal A}, \quad (\nabla_A E(A_W(\beta),\beta W),A' )_{\cal A}=0.
\end{equation}
Using  (\ref{eq:expenergy}) and observing that
\begin{eqnarray}\label{eq:expgamma}
 \Gamma_l(A_W(\beta),A')=\sum_{k\geq l-1}\beta^k\sum_{\underset{|\alpha|_1=k, |\alpha|_\infty<k}{\alpha \in (\N^\ast)^{l-1}} }\sum_{i=1}^{l}\gamma_l(\tau_{(i,l)}(A_W^{(\alpha_1)},\cdots,A_W^{(\alpha_{l-1})},A')),
\end{eqnarray}
we can rewrite (\ref{eq:zerograd}) by collecting the terms of order $\beta^k$ as
$$
\forall k \in \N^\ast, \quad \forall A' \in {\cal A}, \quad \langle 2\Theta(A_W^{(k)})+B_W^{(k)},A' \rangle = 0,
$$ 
where $B_W^{(k)}$ is given by (\ref{eq:formulaB1}) for $k=1$ and by (\ref{eq:formulaBk}) for the general case $k \ge 2$. Thus (\ref{eq:linsysA}) is proved. 

\medskip 

\noindent
Using (\ref{eq:gamma1A}) and (\ref{eq:gammak}), we can rewrite (\ref{eq:EWk}) for $k=2n+\epsilon$ ($n \in \N$, $\epsilon\in\left\{0,1\right\}$) as
\begin{eqnarray*}
{\cal E}_W^{(2n+\epsilon)} &=&   \tr(H_0\gamma_1(A_W^{(2n+\epsilon)})) + \sum_{2\leq l\leq 2n+\epsilon} \;\;\;
    \sum_{\alpha \in (\N^\ast)^l \, | \, |\alpha|_1=2n+\epsilon} \tr(H_0 \gamma_{W,l}^\alpha) \nonumber \\
&+& \frac{1}{2} \sum_{\underset{l,l'\geq 1}{2 \le l+l'\le 2n+\epsilon}} \;\;\; \sum_{{\alpha \in (\N^\ast)^l, \, \alpha' \in (\N^\ast)^{l'} \, | \, |\alpha|_1+|\alpha'|_1 = 2n+\epsilon}} D(\rho_{\gamma_{W,l}^{\alpha}},\rho_{\gamma_{W,l'}^{\alpha'}}) \nonumber \\
    & +& \sum_{2\leq l\leq 2n+\epsilon-1}\;\;\;  \sum_{\alpha \in (\N^\ast)^l \, | \, |\alpha|_1 = 2n+\epsilon-1} \int_{\R^3}\rho_{\gamma_{W,l}^{\alpha}} W
\\ &=& 
\sum_{2\leq l\leq 2n+\epsilon} \;\;\;
    \sum_{\alpha \in (\N^\ast)^l \, | \, |\alpha|_1=2n+\epsilon,\, |\alpha|_\infty\leq n} \tr(H_0 \gamma_{W,l}^\alpha) \nonumber \\
&+& \frac{1}{2} \sum_{\underset{l,l'\geq 1}{2 \le l+l'\le 2n+\epsilon}} \;\;\; \sum_{\underset{\max(|\alpha|_\infty,\,|\alpha'|_\infty)\leq n}{\alpha \in (\N^\ast)^l, \, \alpha' \in (\N^\ast)^{l'} \, | \, |\alpha|_1+|\alpha'|_1 = 2n+\epsilon}} D(\rho_{\gamma_{W,l}^{\alpha}},\rho_{\gamma_{W,l'}^{\alpha'}}) \nonumber \\
    & +& \sum_{2\leq l\leq 2n+\epsilon-1}\;\;\;  \sum_{\alpha \in (\N^\ast)^l \, | \, |\alpha|_1 = 2n+\epsilon-1,\, |\alpha|_\infty\leq n} \int_{\R^3}\rho_{\gamma_{W,l}^{\alpha}} W+J_{2n+\epsilon}(A_W^{(1)},...,A_W^{(2n+\epsilon-1)}),
\end{eqnarray*}
where 
\begin{eqnarray*}
 J_{2n+\epsilon}(A_W^{(1)},...,A_W^{(2n+\epsilon-1)})=& &\displaystyle\sum_{2\leq l\leq 2n+\epsilon}\,\,\,\sum_{\underset{|\alpha|_\infty> n}{\alpha \in (\N^\ast)^l \, | \,|\alpha|_1=2n+\epsilon}} \tr(H_0\gamma_l^{(\alpha)})\\ 
&+&\frac{1}{2}\displaystyle\sum_{\underset{l_1,l_2\geq 1}{2\leq l_1+l_2\leq 2n+\epsilon}}\,\,\, \sum_{\underset{\max(|\alpha|_\infty,|\alpha'|_\infty)> n}{\alpha \in (\N^\ast)^l,\,\alpha' \in (\N^\ast)^{l'} \, | \, |\alpha|_1+|\alpha'|_1=2n+\epsilon}}D(\rho_{\gamma_{l_1}^{(\alpha)}},\rho_{\gamma_{l_2}^{(\alpha')}})\\ 
&+&\displaystyle \sum_{1\leq l\leq 2n+\epsilon-1}\,\,\, \sum_{\underset{|\alpha|_\infty> n}{\alpha \in (\N^\ast)^l \, | \,|\alpha|_1=2n+\epsilon-1}}\int_{\R^3}\rho_{\gamma_l^{(\alpha)}}W.
\end{eqnarray*} 
As
\begin{eqnarray*}
  J_{2n+\epsilon}(A_W^{(1)},...,A_W^{(2n+\epsilon-1)})=\sum_{k=n}^{2n+\epsilon-1} \langle 2\Theta(A_W^{(2n+\epsilon-k)})+B_W^{(2n+\epsilon-k)},A_W^{(k)}\rangle=0,
\end{eqnarray*} 
the proof of the fifth statement is complete. Lastly, the sixth statement can be established reasoning as in the proof of Theorem~\ref{th:Wigner_non_degenerate}.

\bigskip

\noindent
{\bf Ackowledgments.} This work was completed while the authors were core participants to the IPAM program {\it Materials for a sustainable energy future}. Financial support from IPAM and the ANR grant Manif is gracefully acknowledged.

\end{document}